\ifdefined\pdfoutput\else
  \errmessage{Wrong TeX engine: select pdfLaTeX, or run build-paper.cmd}
\fi
\makeatletter
\def\input@path{{template/}}
\makeatother
\documentclass[a4paper,USenglish]{lipics-v2021}
\pdfoutput=1
\hideLIPIcs
\usepackage{mathtools}
\usepackage{bm}
\usepackage{booktabs}
\usepackage{tabularx}
\usepackage{xcolor}
\usepackage{tikz}
\newcommand{\kmain}{k_0}% number of main rows (formerly m)
\newcommand{\Pmod}{\Pi}% CRT modulus (formerly P)
\newcommand{\Ncl}{N_c}% number of clauses (formerly C)
\newcommand{\slotreq}[1]{\gamma_{#1}}% slot requirement (formerly m_i)
\mathtoolsset{showonlyrefs}
\title{Computational Complexity and Exact Algorithms for Pressure-Driven Hydrogen Charging}
\titlerunning{Complexity and Exact Algorithms for Hydrogen Charging}
\author{Nanase Sato}{Nagoya University}{sato.nanase.w5@s.mail.nagoya-u.ac.jp}{}{}
\author{Wei Wu}{Shizuoka University}{goi@shizuoka.ac.jp}{https://orcid.org/0000-0002-0443-3642}{}
\author{Mutsunori Yagiura}{Nagoya University}{yagiura@i.nagoya-u.ac.jp}{https://orcid.org/0000-0003-0970-9414}{}
\authorrunning{N. Sato, W. Wu and M. Yagiura}
\Copyright{Nanase Sato, Wei Wu and Mutsunori Yagiura}
\category{}
\relatedversion{}
\EventEditors{}
\EventNoEds{0}
\EventLongTitle{44th International Symposium on Theoretical Aspects of Computer Science (STACS 2027)}
\EventShortTitle{STACS 2027}
\EventAcronym{STACS}
\EventYear{2027}
\EventDate{March 8--12, 2027}
\EventLocation{G\"ottingen, Germany}
\EventLogo{}
\SeriesVolume{}
\ArticleNo{}

\begin{document}
\nolinenumbers
\maketitle
\begin{abstract}
We study the problem of deciding the feasibility of pressure-driven hydrogen charging from supply cylinders to demand cylinders without replenishment, asking whether an order of transfers allows every demand cylinder to reach its target pressure.
Hydrogen can be transferred from a cylinder to another only when the supply cylinder has higher pressure than the demand cylinder.
The charging to each demand cylinder must be completed before starting another, and each transfer continues until the connected cylinders reach equal pressures or the demand cylinder reaches its target.
We consider equal-capacity supply cylinders and arbitrary positive rational demand cylinder capacities.

Let \(n\) be the number of demand cylinders and \(L\) the binary input length.
With one supply cylinder, feasibility is decidable in polynomial time; with two, it is NP-complete even with cylinder capacities equal to $1$ and nonnegative integer pressures.
We prove that feasibility is preserved even if the supply cylinders are always used in non-decreasing order of their current pressures when filling each demand cylinder.
We also show that majorization of residual supply cylinder pressure states justifies discarding majorized states in dynamic programming over sets of completed demand cylinders without losing any feasible remaining demand cylinder order.
Using these structural results, we design exact algorithms for deciding feasibility and analyze their running times; we also establish complexity bounds for the problem.
With two supply cylinders, one state per completed demand cylinder set suffices, yielding \(2^n\operatorname{poly}(L)\) time and space.
With three, the exponential time hypothesis (ETH) excludes deterministic $2^{o(n\log n)}\operatorname{poly}(L)$-time algorithms.
This lower bound is tight in the exponent because the rule for order of supply cylinders allows feasibility to be decided in \(n!\operatorname{poly}(L)\) time by enumerating only demand cylinder orders.
\end{abstract}

%Enumerating demand cylinder orders takes $n!\operatorname{poly}(L)$ time.

%%%%%%%%%%%%%%%%%%%%%%%%%%%%%%%%%%%%%%%%%%%%%%%%%%%%%%%%%%%%%%%%%%%%
\clearpage

\section{Introduction}
\label{sec:introduction}
Hydrogen refueling stations fill vehicle tanks using several supply cylinders at different pressures.
Research on these systems has examined filling control and storage configurations, focusing on refueling time and energy consumption~\cite{bai2021}.
A common filling method uses lower-pressure supply cylinders before higher-pressure ones.
Studies have also examined when to switch between these cylinders to improve refueling performance~\cite{caponi2022,oh2024}.

We study the problem of deciding the feasibility of pressure-driven hydrogen charging under an isothermal ideal-gas model without replenishment.
Each cylinder has a specified capacity and initial pressure; each demand cylinder also has a target pressure.
To reflect station operations in which vehicles are filled one at a time~\cite{baetge2026,brown2022,li2021}, we require each demand cylinder to reach its target pressure before starting to charge another demand cylinder.
Hydrogen flows only from a higher-pressure supply cylinder to a lower-pressure demand cylinder.
Each transfer continues until the pair reaches equal pressures or the demand cylinder reaches its target pressure, whichever occurs first.
We ask whether a transfer sequence fills every demand cylinder to its target.

Completing the same set of demand cylinders consumes a fixed total amount of hydrogen, but different filling orders can leave different supply cylinder pressure distributions.
These distributions affect whether all remaining demand cylinders can reach their target pressures.
This dependence on the residual pressure distribution complicates this problem.

%%%%%%%%%%%%%%%%%%%%%%%%%%%%%%%%%%%%%%%%%%%%%%%%%%%%%%%%%%%%%%%%%%%%
\paragraph*{Our results.}
We develop a state reduction method based on majorization for the feasibility problem.
We also establish complexity bounds for different numbers of supply cylinders.

For equal-capacity supply cylinders, if a demand cylinder can be filled from a supply cylinder state, using supply cylinders in non-decreasing order of their current pressures also succeeds.
We prove this fact and that this rule leaves a residual supply cylinder pressure vector majorizing that of every other successful supply cylinder order.
Majorization compares how unevenly a fixed total residual pressure is distributed~\cite{marshall2011}.
We also prove that a majorizing state preserves every feasible remaining demand cylinder order, with majorization maintained after each demand cylinder is filled.
These results justify searching only over demand cylinder orders and discarding majorized states.

Based on these findings, we show the following.
Let \(n\) be the number of demand cylinders and \(L\) the binary input length.
With one supply cylinder, sorting by non-increasing target pressure decides feasibility in polynomial time.
With two supply cylinders, conservation fixes the residual pressure total for each completed demand cylinder set; majorization totally orders its residual states.
The dynamic program retains at most one representative state for each of the $2^n$ possible sets of completed demand cylinders, using $2^n\operatorname{poly}(L)$ time and space.
Even when all cylinders have unit capacity, the problem is NP-complete for nonnegative integer pressures, and the exponential time hypothesis (ETH) rules out algorithms running in $2^{o(n)}\operatorname{poly}(L)$ time.
The problem is strongly NP-complete for rational pressures even when all cylinders have unit capacity.

With three supply cylinders, ETH excludes deterministic $2^{o(n\log n)}\operatorname{poly}(L)$-time algorithms, even with unit-capacity supply cylinders and positive integer pressures.
By the supply cylinder ordering result above, enumerating demand cylinder orders yields an \(n!\operatorname{poly}(L)=2^{O(n\log n)}\operatorname{poly}(L)\)-time algorithm, establishing the optimal deterministic exponent \(\Theta(n\log n)\) under ETH.
Table~\ref{tab:results} summarizes our time bounds and hardness for the cases with $m = 1, 2, 3$ supply cylinders.
\begin{table}[!t]
\centering\small
\caption{Our running-time upper bounds and hardness results for charging feasibility problems with $m$ equal-capacity supply cylinders and arbitrary positive rational demand cylinder capacities.}
\label{tab:results}
\begin{tabularx}{\linewidth}{@{}l>{\raggedright\arraybackslash}p{3.0cm}>{\raggedright\arraybackslash}p{4.1cm}>{\raggedright\arraybackslash}X@{}}
\toprule
$m$ & Our upper bound & ETH lower bound\newline (deterministic algorithms) & NP-completeness\\
\midrule
$1$ & $O(n\log n)$ operations\newline (Thm.~\ref{thm:one-supply}) & --- & --- \\
\addlinespace
$2$ & $2^n\operatorname{poly}(L)$ time\newline (Thm.~\ref{thm:two-supply-dp}) & No $2^{o(n)}\operatorname{poly}(L)$ time\newline (Cor.~\ref{cor:target-unit-eth}) & NP-complete (Thm.~\ref{thm:two-supply-unit})\newline Strongly NP-complete\newline (unary rational pressures; Cor.~\ref{cor:rational-strong-np})\\
\addlinespace
$3$ & $n!\operatorname{poly}(L)$ time\newline (Thm.~\ref{thm:fixed-order-standardization}) & No $2^{o(n\log n)}\operatorname{poly}(L)$ time\newline (Thm.~\ref{thm:three-supply-eth}) & NP-complete (Thm.~\ref{thm:two-supply-unit} instance plus a\newline supply cylinder at initial pressure $0$) \\
\bottomrule
\end{tabularx}
\end{table}

%%%%%%%%%%%%%%%%%%%%%%%%%%%%%%%%%%%%%%%%%%%%%%%%%%%%%%%%%%%%%%%%%%%%
\paragraph*{Related work.}
Scheduling with non-renewable resources studies job sequencing under resource constraints, with each job's requirements specified in advance as fixed constants~\cite{bentert2023,berczi2024,carlier1982}.
In our model, on the other hand, even though each demand cylinder's total hydrogen requirement is fixed, withdrawals from individual supply cylinders can vary depending on current pressures and the order of supply cylinders.
Thus, completing the same demands in different orders can leave different residual distributions, affecting feasibility.

%%%%%%%%%%%%%%%%%%%%%%%%%%%%%%%%%%%%%%%%%%%%%%%%%%%%%%%%%%%%%%%%%%%%
\section{Model and Basic Properties}

%%%%%%%%%%%%%%%%%%%%%%%%%%%%%%%%%%%%%%%%%%%%%%%%%%%%%%%%%%%%%%%%%%%%
\subsection{Supply Cylinder Model and Transfers}
Let $\mathcal S=\{1,\ldots,m\}$ and $D=\{1,\ldots,n\}$ be the sets of supply cylinders and demand cylinders, respectively.
Supply cylinders have a capacity and initial pressure; demand cylinders also have a target pressure:
\begin{align*}
 \text{Supply cylinder }i\in\mathcal S\colon\quad &(v_i^{(S)},\bar p_i^{(S)}),\quad v_i^{(S)}>0,\quad \bar p_i^{(S)}\ge0,\\
 \text{Demand cylinder }j\in D\colon\quad &(v_j^{(D)},\bar p_j^{(D)},\widetilde p_j^{(D)}),\quad v_j^{(D)}>0,\quad 0\le\bar p_j^{(D)}\le\widetilde p_j^{(D)}.
\end{align*}
All capacities and pressures are binary-encoded rationals, with total binary input length \(L\).
Under the isothermal ideal-gas assumption, the amount of hydrogen is proportional to the product of capacity and pressure.
Only transfers from supply cylinders to demand cylinders are allowed, with no replenishment or leakage.
Let \(p_i^{(S)}\) and \(p_j^{(D)}\) be the current pressures, initially \(\bar p_i^{(S)}\) and \(\bar p_j^{(D)}\).
Our model requires each transfer to continue until the two pressures are equal or the demand cylinder reaches its target, increasing demand cylinder pressure by
\[
 \Delta p_{ij}^{(D)}:=\min\left\{\max\left\{
 \frac{v_i^{(S)}(p_i^{(S)}-p_j^{(D)})}{v_i^{(S)}+v_j^{(D)}},0\right\},
 \widetilde p_j^{(D)}-p_j^{(D)}\right\},
\]
and the pressure update is
\begin{equation}
 (p_i^{(S)},p_j^{(D)})\longmapsto
 \left(p_i^{(S)}-\frac{v_j^{(D)}}{v_i^{(S)}}\Delta p_{ij}^{(D)},\,
 p_j^{(D)}+\Delta p_{ij}^{(D)}\right).
 \label{eq:charging-update}
\end{equation}
Only transfers with \(\Delta p_{ij}^{(D)}>0\) are legal.
Stopping at a demand cylinder's target is called \emph{target stopping}.

%%%%%%%%%%%%%%%%%%%%%%%%%%%%%%%%%%%%%%%%%%%%%%%%%%%%%%%%%%%%%%%%%%%%
\subsection{Schedules and the Decision Problem}
\label{sec:schedules-and-problems}

A \emph{schedule} is a finite sequence of legal transfers \(\pi=((i_1,j_1),\ldots,(i_u,j_u))\), with \(|\pi|=u\); pair \((i_k,j_k)\) specifies a transfer from supply cylinder \(i_k\) to demand cylinder \(j_k\).
A \emph{filling} comprises the transfers to one demand cylinder, which is \emph{completed} upon reaching its target.
Each demand cylinder must be completed before another is filled; supply cylinders may be switched or skipped.
A schedule completing every demand cylinder is \emph{feasible}.
The \emph{target hydrogen charging problem} (THCP) asks whether such a schedule exists.
For example, take equal-capacity cylinders with supply cylinder pressures $(4,10)$ and demand cylinder (initial, target) pressures $(0,5)$ and $(0,4)$.
Table~\ref{tab:order-example} compares three filling orders.
\par\smallskip
\noindent\begin{minipage}{\linewidth}
\centering\small
\captionsetup{hypcap=false}
\captionof{table}{Supply states after the first fill; default supply order: $1,2$.}
\label{tab:order-example}
\begin{tabularx}{\linewidth}{@{}lcc>{\raggedright\arraybackslash}X@{}}
\toprule
\shortstack[l]{Demand\\cylinder order} & \shortstack{Supply order\\for first fill} & Supply cylinder state & Second demand cylinder pressure\\
\midrule
$1,2$ & $1\to2$ & $(2,7)$ & $4$ (target reached)\\
$2,1$ & $1\to2$ & $(2,8)$ & $9/2<5$\\
$1,2$ & $2$ only & $(4,5)$ & at most $7/2<4$ (either order)\\
\bottomrule
\end{tabularx}
\end{minipage}
\par\smallskip

%%%%%%%%%%%%%%%%%%%%%%%%%%%%%%%%%%%%%%%%%%%%%%%%%%%%%%%%%%%%%%%%%%%%
\subsection{Basic Properties}
\label{sec:basic-properties}

%%%%%%%%%%%%%%%%%%%%%%%%%%%%%%%%%%%%%%%%%%%%%%%%%%%%%%%%%%%%%%%%%%%%
\paragraph*{Conservation.}
Update~\eqref{eq:charging-update} preserves the capacity-weighted pressure sum.
Thus, throughout any schedule,
\[
 \sum_{i\in\mathcal S}v_i^{(S)}p_i^{(S)}+\sum_{j\in D}v_j^{(D)}p_j^{(D)}
 =\sum_{i\in\mathcal S}v_i^{(S)}\bar p_i^{(S)}+\sum_{j\in D}v_j^{(D)}\bar p_j^{(D)}.
\]

%%%%%%%%%%%%%%%%%%%%%%%%%%%%%%%%%%%%%%%%%%%%%%%%%%%%%%%%%%%%%%%%%%%%
\paragraph*{Pressure monotonicity.}
Supply cylinder pressures never increase, and demand cylinder pressures never decrease.
A pair cannot be reused after equilibration or demand cylinder completion; hence \(|\pi|\le mn\).

\begin{lemma}
\label{lem:exact-simulation}
For binary rational inputs, every pressure in a legal schedule has polynomial bit length.
A proposed schedule of at most $mn$ transfers can be verified in polynomial time.
Consequently, THCP belongs to $\mathrm{NP}$.
\end{lemma}
\begin{proof}
The equilibration and target-stopping formulas in~\eqref{eq:charging-update} are affine, with input-determined rational coefficients and constants of polynomial bit length.
Let $\Gamma$ be the product of the denominators of all input rationals, coefficients, and constants.
There are $O(mn)$ formulas; hence $\log_2\Gamma=\operatorname{poly}(m,n,L)$.
Initially, pressures have denominator $\Gamma$; update coefficients and constants have denominators dividing $\Gamma$.
A transfer changes pressures with denominator $\Gamma^k$ to pressures with denominator $\Gamma^{k+1}$; after $u$ transfers their denominator is $\Gamma^{u+1}$.
Because $u\le mn$ and pressures remain bounded by the largest input pressure, numerators and denominators have polynomial bit length.
A feasible schedule is therefore a polynomial-size certificate verifiable by exact simulation.
\end{proof}

%%%%%%%%%%%%%%%%%%%%%%%%%%%%%%%%%%%%%%%%%%%%%%%%%%%%%%%%%%%%%%%%%%%%
\section{Complexity with One or Two Supply Cylinders}
\label{sec:complexity}

%%%%%%%%%%%%%%%%%%%%%%%%%%%%%%%%%%%%%%%%%%%%%%%%%%%%%%%%%%%%%%%%%%%%
\begin{theorem}
\label{thm:one-supply}
With one supply cylinder, THCP is feasible if and only if a demand cylinder order of non-increasing target pressure is feasible.
This condition can be checked using \(O(n\log n)\) rational arithmetic operations and comparisons.
\end{theorem}

\begin{proof}
Discard demand cylinders already at their targets.
Set $w_j:=v_j^{(D)}(\widetilde p_j^{(D)}-\bar p_j^{(D)})/v_1^{(S)}$.
For a demand cylinder order, let $C_j$ be the cumulative pressure decrease of the supply cylinder through $j$.
By~\eqref{eq:charging-update}, the order is feasible exactly when $C_j\le\bar p_1^{(S)}-\widetilde p_j^{(D)}$ for every $j$.
Interpreting demand cylinders as jobs, their processing times are $w_j$ and deadlines are $\bar p_1^{(S)}-\widetilde p_j^{(D)}$.
Jackson's earliest-due-date rule~\cite{jackson1955} therefore implies that the order of non-increasing target pressure is feasible whenever some order is.
Sorting and checking prefix sums take $O(n\log n)$ operations.
\end{proof}

%%%%%%%%%%%%%%%%%%%%%%%%%%%%%%%%%%%%%%%%%%%%%%%%%%%%%%%%%%%%%%%%%%%%
\begin{theorem}
\label{thm:two-supply-unit}
With two supply cylinders and all capacities equal to one, THCP is NP-complete even for nonnegative integer pressures.
\end{theorem}

\begin{proof}
Lemma~\ref{lem:exact-simulation} establishes NP membership.
We reduce from \textsc{Partition}.
Let \(a_1,\ldots,a_k\) be the positive integers in the input instance.
If their sum is odd, the reduction returns a fixed no-instance.
Otherwise, let \(\sum_{i=1}^k a_i=2B\).
The instance has unit-capacity supply cylinders \(s_1,s_2\) at pressures \(7B,11B\), ordinary demand cylinders \(d_i=(6B-a_i,6B)\), and a gate demand cylinder \(d_*=(0,7B)\), specified by (initial, target) pressures.
It has \(k+1\) demand cylinders, all with unit capacity; construction takes polynomial time and uses nonnegative integer pressures.
Completing demand cylinder \(d_i\) decreases the sum of supply cylinder pressures by \(a_i\).
We show that completing \(d_*\) requires using \(s_1\) at pressure \(6B\), followed by \(s_2\) at its initial pressure \(11B\).
Thus, the ordinary demand cylinders completed before \(d_*\) use only \(s_1\) and have total weight \(B\).

%%%%%%%%%%%%%%%%%%%%%%%%%%%%%%%%%%%%%%%%%%%%%%%%%%%%%%%%%%%%%%%%%%%%
\paragraph*{Completeness (\(\Leftarrow\)).}
Suppose \(I\subseteq\{1,\ldots,k\}\) satisfies \(\sum_{i\in I}a_i=B\).
Figure~\ref{fig:partition-completeness} shows a feasible schedule: each group of ordinary demand cylinders has total weight $B$ and uses one supply cylinder, whose pressure remains at least $6B$ after every fill, in any order.
\par\smallskip
\noindent\begin{minipage}{\linewidth}
\centering\small
\begin{tikzpicture}[x=1cm,y=1cm,font=\small,>=stealth]
\node (initial) at (0,0) {$(7B,11B)$};
\node (beforegate) at (3.5,0) {$(6B,11B)$};
\node (aftergate) at (7.5,0) {$(3B,7B)$};
\node (final) at (11,0) {$(3B,6B)$};
\draw[->] (initial) -- node[above] {$d_i,\ i\in I$} node[below] {using $s_1$} (beforegate);
\draw[->] (beforegate) -- node[above] {$d_*:\ s_1\to s_2$} node[below] {$0\to3B\to7B$} (aftergate);
\draw[->] (aftergate) -- node[above] {$d_i,\ i\notin I$} node[below] {using $s_2$} (final);
\end{tikzpicture}
\captionsetup{hypcap=false}
\captionof{figure}{Pressures of $s_1,s_2$; pressure of $d_*$ below the middle arrow.}
\label{fig:partition-completeness}
\end{minipage}
\par\smallskip

%%%%%%%%%%%%%%%%%%%%%%%%%%%%%%%%%%%%%%%%%%%%%%%%%%%%%%%%%%%%%%%%%%%%
\paragraph*{Soundness ($\Rightarrow$).}
Suppose the THCP instance is feasible.
Neither supply cylinder alone can complete \(d_*\), and the order \(s_2,s_1\) raises its pressure to at most \(25B/4<7B\).
Hence, filling \(d_*\) requires transfers from \(s_1\) followed by \(s_2\).
Let $x$ and $y$ be the pressures of $s_1$ and $s_2$, respectively, immediately before they connect to $d_*$ in that order.
After equilibration with $s_1$, both $d_*$ and $s_1$ have pressure $x/2$.
Because equilibration with \(s_2\) would yield \((y+x/2)/2\), completing \(d_*\) requires \((y+x/2)/2\ge7B\).
Because \(y\) cannot exceed \(s_2\)'s initial pressure \(11B\), this implies \(x\ge6B\).

Because \(a_i\le2B\) implies \(x/2\le7B/2\le6B-a_i\), \(s_1\) cannot transfer hydrogen to any further demand cylinder and ends at \(x/2\).
After completing the gate demand cylinder at target \(7B\), \(s_2\) has pressure at least \(7B\).
If unused thereafter, this is its final bound. Otherwise, each transfer must complete its demand cylinder, because equilibration below \(6B\) leaves it unfinished.
Thus \(s_2\) ends at pressure at least \(6B\).

Completing all demand cylinders raises their total pressure by \(2B+7B=9B\), leaving final supply cylinder pressure sum \(18B-9B=9B\) by conservation.
Because \(s_1\) ends at \(x/2\) and \(s_2\) ends at pressure at least \(6B\), we have \(x/2+6B\le9B\), implying \(x\le6B\).
Together with \(x\ge6B\), this yields \(x=6B\).
The gate condition \((y+x/2)/2\ge7B\) then implies \(y\ge11B\).
Because \(y\le11B\), we conclude that \(y=11B\).
Because \(y=11B\) and supply cylinder pressures never increase, \(s_2\) was not used before \(d_*\).
The demand cylinders preceding \(d_*\) were therefore completed using only \(s_1\) and have total weight \(7B-6B=B\), yielding a solution to \textsc{Partition}.
\end{proof}

\begin{corollary}
\label{cor:rational-strong-np}
THCP is strongly NP-complete even with exactly two supply cylinders and all capacities equal to one.
This holds for nonnegative rational pressures whose reduced numerators and denominators are encoded in unary.
\end{corollary}

\begin{proof}
The reduction in Theorem~\ref{thm:two-supply-unit} requires only positive weights summing to $2B$ and therefore also applies to rational \textsc{Partition}.
Wojtczak~\cite[proofs of Theorems~4 and~5]{wojtczak2018} reduces bounded-occurrence \textsc{All-the-Same-SAT} with $r$ variables to rational \textsc{Partition} on $2r$ positive weights summing to $2r$.
Their reduced numerators and denominators are polynomially bounded in the source input size.
Apply Theorem~\ref{thm:two-supply-unit} with $B=r$ and let $a_i=u_i/v_i$ be in lowest terms.
The only potentially nonintegral pressures are $6r-a_i=(6rv_i-u_i)/v_i$, with positive reduced numerators at most $6rv_i$ and denominators at most $v_i$.
All other pressures are integers bounded by $11r$, and there are $2r+1$ demand cylinders.
The unary output length is polynomial, proving strong NP-hardness; NP membership follows from Lemma~\ref{lem:exact-simulation}.
\end{proof}

\begin{corollary}
\label{cor:target-unit-eth}
Under ETH, THCP admits no \(2^{o(n)}\operatorname{poly}(L)\)-time algorithm, even with two supply cylinders, unit capacities, and nonnegative integer pressures.
\end{corollary}

\begin{proof}
Under ETH, 3SAT admits no \(2^{o(q)}\)-time algorithm, where \(q\) is the total number of variables and clauses~\cite[Theorem~8]{kpw2019}.
The standard reduction produces an equivalent \textsc{Subset Sum} instance with \(k=O(q)\) nonnegative weights and a target \(t\), all with \(O(q)\) bits~\cite[Section~3.3]{kpw2019}.
For the standard reduction to \textsc{Partition}~\cite[p.~100]{karp1972}, let \(W_{\rm tot}\) be the total weight.
Remove zero weights, handle targets outside \((0,W_{\rm tot})\) in polynomial time, and append \(t+1\) and \(W_{\rm tot}-t+1\).
Each half must sum to \(W_{\rm tot}+1\), forcing the original weights alongside \(W_{\rm tot}-t+1\) to sum to \(t\).
The resulting equivalent instance has at most \(k+2=O(q)\) positive weights, each with \(O(q)\) bits.

Theorem~\ref{thm:two-supply-unit} produces an equivalent THCP instance with one demand cylinder per weight and one gate demand cylinder: \(n=O(q)\), pressures of \(O(q)\) bits, and \(L=O(q^2)\).
A \(2^{o(n)}\operatorname{poly}(L)\)-time algorithm would then solve 3SAT in \(2^{o(q)}\) time, contradicting ETH.
\end{proof}

%%%%%%%%%%%%%%%%%%%%%%%%%%%%%%%%%%%%%%%%%%%%%%%%%%%%%%%%%%%%%%%%%%%%
\section{Cascade Dominance}
\label{sec:cascade-dominance}

Throughout this section, all supply cylinders have the same capacity $v^{(S)}$.

%%%%%%%%%%%%%%%%%%%%%%%%%%%%%%%%%%%%%%%%%%%%%%%%%%%%%%%%%%%%%%%%%%%%
\subsection{Majorization and Cascade Operators}
\label{subsec:majorization-operators}

%%%%%%%%%%%%%%%%%%%%%%%%%%%%%%%%%%%%%%%%%%%%%%%%%%%%%%%%%%%%%%%%%%%%
\paragraph*{Majorization.}
We represent residual supply cylinder states in non-decreasing order.
For sorted states $P=(p_1,\ldots,p_m)$ and $Q=(q_1,\ldots,q_m)$, we use the standard majorization order~\cite{marshall2011}:
\[
 P\succeq_{\mathrm{maj}}Q\quad\Longleftrightarrow\quad
 \sum_{i=1}^m p_i=\sum_{i=1}^m q_i
 \quad\text{and}\quad
 \sum_{i=1}^h p_i\le\sum_{i=1}^h q_i\quad(1\le h<m).
\]
In Section~\ref{sec:schedules-and-problems}, the supply cylinder states $(2,7)$ and $(4,5)$ after filling demand cylinder 1 have equal totals and are componentwise incomparable. However, $(2,7)\succeq_{\mathrm{maj}}(4,5)$ because $2\le4$, and only $(2,7)$ permits completing demand cylinder 2.
Theorems~\ref{thm:single-demand-dominance} and~\ref{thm:propagation} establish the general preference: filling a demand cylinder in non-decreasing supply cylinder pressure order yields a residual state majorizing that of any other successful order. Replacing the initial state by a majorizing state preserves feasibility under this rule.
For a real threshold $\lambda$ and $S=(x_1,\ldots,x_m)$, define $(a)_+:=\max\{a,0\}$ and $\Phi_\lambda(S):=\sum_i(x_i-\lambda)_+$, the sum of pressure excesses above $\lambda$.
For any states \(P\) and \(Q\) with the same sum,
\begin{equation}
 P\succeq_{\mathrm{maj}}Q\quad\Longleftrightarrow\quad
 \Phi_\lambda(P)\ge\Phi_\lambda(Q)\quad(\forall\lambda\in\mathbb R).
 \label{eq:maj-phi}
\end{equation}
This is a standard characterization of majorization~\cite[Proposition~4.B.3]{marshall2011}.

%%%%%%%%%%%%%%%%%%%%%%%%%%%%%%%%%%%%%%%%%%%%%%%%%%%%%%%%%%%%%%%%%%%%
\paragraph*{Cascade filling.}
Cascade filling scans supply cylinders in non-decreasing current pressure order to fill one demand cylinder.
It skips cylinders at or below the demand cylinder pressure and stops at the target or after scanning all cylinders.
Define the capacity ratio $r_j:=v_j^{(D)}/v^{(S)}$ and let $\alpha_j:=r_j/(1+r_j)$ and $\beta_j:=1/(1+r_j)$.
For a sorted supply cylinder state $S=(x_1,\ldots,x_m)$, the demand cylinder pressure update at supply cylinder $i$ before target stopping is
\[
 p_j^{(D)}\longmapsto\begin{cases}p_j^{(D)},&x_i\le p_j^{(D)},\\
 \alpha_j p_j^{(D)}+\beta_jx_i,&x_i>p_j^{(D)}
 \end{cases}.
\]
The update is non-decreasing in $p_j^{(D)}$.
Let \(F_j(S)\) denote the residual supply cylinder state after cascade filling of demand cylinder \(j\).
For a supply cylinder permutation \(\sigma\), let \(G_{(j,\sigma)}(S)\) denote the sorted residual state after filling demand cylinder \(j\) in \(\sigma\)-order.
Either operator returns \(\bot\) if its filling procedure fails.
Whenever both orders complete demand cylinder $j$, their residual supply cylinder pressure sums agree by conservation.
Pressure monotonicity prevents reuse within a fill; thus appending unused cylinders extends any successful fill to a permutation with the same residual state.

\begin{lemma}
\label{lem:cascade-supply-order}
Suppose equal-capacity supply cylinders are indexed by non-decreasing pressure at the start of a fill.
Cylinders with the same pressure at that time may be indexed in either order.
Cascade filling leaves their residual pressures non-decreasing by index.
\end{lemma}
\begin{proof}
If no supply cylinder is used, the supply cylinder state is unchanged.
Otherwise, the cylinders skipped at the start keep their pressures unchanged.
Each equilibration below the target leaves the supply cylinder at the new demand cylinder pressure.
Because the demand cylinder pressure increases with each transfer, the residual pressures of these cylinders are non-decreasing in the order of use.
If supply cylinder \(i\) completes demand cylinder \(j\), its residual pressure is at least the target \(\widetilde p_j^{(D)}\) and at most its pressure before the transfer, \(p_i^{(S)}\).
Earlier cylinders have residual pressures below the target, and unused later cylinders have pressures at least \(p_i^{(S)}\).
Thus the residual supply cylinder pressures remain non-decreasing.
\end{proof}

%%%%%%%%%%%%%%%%%%%%%%%%%%%%%%%%%%%%%%%%%%%%%%%%%%%%%%%%%%%%%%%%%%%%
\subsection{Cascade Dominance for One Demand Cylinder}
\label{sec:cascade-exchange}
Interchange arguments based on majorization are used in stochastic scheduling~\cite{changyao1993}.
We prove cascade dominance by exchanging adjacent supply cylinders to place the one with lower pressure first.
Such an exchange can raise the demand cylinder pressure after the pair and reduce the pair's total residual pressure.
A higher demand cylinder pressure requires less hydrogen from the remaining cylinders.
The next lemma quantifies the compensation via \(\Phi_\lambda\).

%%%%%%%%%%%%%%%%%%%%%%%%%%%%%%%%%%%%%%%%%%%%%%%%%%%%%%%%%%%%%%%%%%%%
\begin{lemma}
\label{lem:tail-compensation}
Fix a demand cylinder \(j\) with capacity ratio \(r_j:=v_j^{(D)}/v^{(S)}\) and an ordered sequence \(\tau\) of distinct equal-capacity supply cylinders.
Compare two fills from starting pressures \(0\le t_1\le t_2\le\widetilde p_j^{(D)}\), fixing the demand cylinder capacity and target, initial supply cylinder pressures, and order \(\tau\).
Let \(R(t)\) denote the residual supply cylinder pressure vector after filling \(j\) from pressure \(t\), with unused cylinders unchanged.
If filling from \(t_1\) completes \(j\), then filling from \(t_2\) also completes it and $\Phi_\lambda(R(t_2))-\Phi_\lambda(R(t_1)) \ge r_j\bigl((t_2-\lambda)_+-(t_1-\lambda)_+\bigr)$ for every \(\lambda\in\mathbb R\).
\end{lemma}

\begin{proof}
The demand cylinder pressure update is non-decreasing in incoming pressure. Higher pressure requires less hydrogen to reach the target and permits less transfer before equilibration.
Each supply cylinder therefore transfers no more from \(t_2\) than from \(t_1\); along \(\tau\), this implies completion from \(t_2\) and \(R(t_2)\ge R(t_1)\) componentwise.
\begin{itemize}
\item If \(\lambda\le t_1\), cylinders initially below \(\lambda\) are skipped in both fills, whereas all others retain pressure at least \(\lambda\).
Completion from \(t_1\) removes \(r_j(t_2-t_1)\) more total supply cylinder pressure than from \(t_2\), yielding \(\Phi_\lambda(R(t_2))-\Phi_\lambda(R(t_1))=r_j(t_2-t_1)\).
\item If \(\lambda\ge t_2\), the right-hand side is zero.
The componentwise inequality \(R(t_2)\ge R(t_1)\) established above implies \(\Phi_\lambda(R(t_2))\ge\Phi_\lambda(R(t_1))\), proving the claimed inequality.
\item If \(t_1<\lambda<t_2\), filling from \(\lambda\) also completes the demand cylinder, with \(R(\lambda)\ge R(t_1)\) componentwise.
By monotonicity and the first case with starting pressures \(\lambda\) and \(t_2\), we have \(\Phi_\lambda(R(t_2))-\Phi_\lambda(R(t_1))\ge\Phi_\lambda(R(t_2))-\Phi_\lambda(R(\lambda))=r_j(t_2-\lambda)\). \qedhere
\end{itemize}
\end{proof}

%%%%%%%%%%%%%%%%%%%%%%%%%%%%%%%%%%%%%%%%%%%%%%%%%%%%%%%%%%%%%%%%%%%%
\begin{theorem}
\label{thm:single-demand-dominance}
Suppose all supply cylinders have equal capacity and an order \(\sigma\) completes demand cylinder \(j\) from state \(S=(x_1,\ldots,x_m)\).
Then cascade filling from \(S\) also completes \(j\), and its residual state satisfies \(F_j(S)\succeq_{\mathrm{maj}}G_{(j,\sigma)}(S)\).
\end{theorem}

\begin{proof}
In a successful fill, let $a<b$ be the pressures of an adjacent inverted pair, $t$ the demand cylinder pressure before the pair, and $T=\widetilde p_j^{(D)}$.
Let $t_\uparrow,t_\downarrow$ be the demand cylinder pressures after the pair and $L_\uparrow,L_\downarrow$ its sorted residual supply cylinder pressure vectors, for ascending and descending order, respectively.
Exchange the pair as in Figure~\ref{fig:cascade-exchange}.
\par\smallskip
\noindent\begin{minipage}{\linewidth}
\centering\small
\begin{tikzpicture}[x=1cm,y=1cm,font=\small,>=stealth]
\node[anchor=west] at (0,0) {Descending};
\node[anchor=west] at (0,-.6) {Ascending};
\node (pd) at (2.9,0) {same prefix};
\node (pu) at (2.9,-.6) {same prefix};
\node[draw,minimum width=.45cm,inner sep=2pt] (bd) at (5,0) {$b$};
\node[draw,minimum width=.45cm,inner sep=2pt] (ad) at (6.2,0) {$a$};
\node[draw,minimum width=.45cm,inner sep=2pt] (au) at (5,-.6) {$a$};
\node[draw,minimum width=.45cm,inner sep=2pt] (bu) at (6.2,-.6) {$b$};
\node (sd) at (8.7,0) {same suffix};
\node (su) at (8.7,-.6) {same suffix};
\draw[->] (pd) -- node[above,inner sep=1pt] {$t$} (bd);
\draw[->] (bd) -- (ad);
\draw[->] (ad) -- node[above,inner sep=1pt] {$t_\downarrow$} (sd);
\draw[->] (pu) -- node[above,inner sep=1pt] {$t$} (au);
\draw[->] (au) -- (bu);
\draw[->] (bu) -- node[above,inner sep=1pt] {$t_\uparrow$} (su);
\end{tikzpicture}
\captionsetup{hypcap=false,skip=3pt}

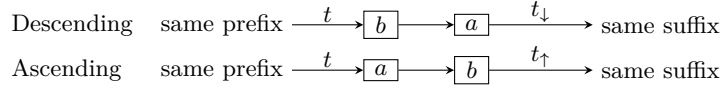
\captionof{figure}{Supply cylinder exchange; the common suffix compensates pair losses via Lemma~\ref{lem:tail-compensation}.}
\label{fig:cascade-exchange}
\end{minipage}
\par\smallskip

%\subparagraph*{Step 1: completion.}
First, if \(t=T\) or \(a\le t\), the exchange changes neither fill; assume \(t<T\) and \(t<a<b\).
Let \(q_a=(r_jt+a)/(1+r_j)\) and \(q_b=(r_jt+b)/(1+r_j)\) be the first-connection equilibrium pressures.
Without target stopping, the demand cylinder pressures after the pair satisfy
\[
 \begin{gathered}
 \widehat t_\uparrow=\frac{r_jq_a+b}{1+r_j},\qquad
 \widehat t_\downarrow=q_b+\frac{(a-q_b)_+}{1+r_j},\qquad
 \widehat t_\uparrow-\widehat t_\downarrow
 =\frac{\min\{b-a,r_j(a-t)\}}{(1+r_j)^2}\ge0.
 \end{gathered}
\]
Stopping at \(T\) replaces each pressure by its minimum with \(T\); hence \(t_\uparrow\ge t_\downarrow\).
By Lemma~\ref{lem:tail-compensation}, the ascending order also completes demand cylinder \(j\).

Next, both orders complete demand cylinder \(j\), leaving equal supply cylinder pressure sums by conservation. The unchanged prefix contributes nothing to the final \(\Phi_\lambda\) difference.
Let \(\Delta_\lambda\) be this difference, measured as ascending minus descending, and define \(\Psi_\lambda(L,t'):=\Phi_\lambda(L)+r_j(t'-\lambda)_+\).
Applying Lemma~\ref{lem:tail-compensation} to the common suffix from demand cylinder pressures \(t_\uparrow\ge t_\downarrow\) yields
$\Delta_\lambda\ge\Phi_\lambda(L_\uparrow)-\Phi_\lambda(L_\downarrow)+r_j\bigl((t_\uparrow-\lambda)_+-(t_\downarrow-\lambda)_+\bigr)=\Psi_\lambda(L_\uparrow,t_\uparrow)-\Psi_\lambda(L_\downarrow,t_\downarrow).$
By~\eqref{eq:maj-phi}, it suffices to prove that this \(\Psi_\lambda\) difference is nonnegative for every \(\lambda\in\mathbb R\).

%\subparagraph*{Step 3: smaller pair residuals.}
To bound \(\Delta_\lambda\), we compare the smaller residual supply cylinder pressures.
Let \(p_\uparrow=\min L_\uparrow\), \(p_\downarrow=\min L_\downarrow\), and \(w_{\mathrm{rem}}=r_j(T-t)\), the required total supply cylinder pressure decrease from \(t\).
In the ascending order, the lower-pressure cylinder ends at \(a-w_{\mathrm{rem}}\) if it completes the demand cylinder, and at \(q_a\) otherwise.
It has the smaller residual in either case; hence \(p_\uparrow=\max\{q_a,a-w_{\mathrm{rem}}\}\).
In the descending order, the higher-pressure cylinder ends at pressure at least \(q_b\).
The lower-pressure cylinder remains at \(a\) if unused, and ends at pressure at least \(q_b\) if used.
Because \(a>q_a\) and \(q_b>q_a\), we obtain \(p_\downarrow\ge q_a\).
Each cylinder's pressure decrease is at most \(w_{\mathrm{rem}}\), and both initial pressures are at least \(a\).
Thus \(p_\downarrow\ge a-w_{\mathrm{rem}}\), and combining these bounds yields
\(p_\downarrow\ge\max\{q_a,a-w_{\mathrm{rem}}\}=p_\uparrow\).

%\subparagraph*{Step 4: the $\Psi_\lambda$ inequality.}
Finally, we use this comparison to prove \(\Delta_\lambda\ge0\) for every \(\lambda\in\mathbb R\).
In either order, the larger residual supply cylinder pressure is at least the demand cylinder pressure, with equality if filling remains unfinished.
\begin{itemize}
\item For \(\lambda\le t_\downarrow\), let \(L=(p_1,p_2)\), with \(p_1\le p_2\), be the pair residual state in either order.
Let \(t'\) be the corresponding demand cylinder pressure after the pair.
Because \(p_2\ge t'\ge\lambda\), we have \((p_2-\lambda)_+=p_2-\lambda\) and \((t'-\lambda)_+=t'-\lambda\).
The smaller residual \(p_1\) may lie below \(\lambda\), and its term remains \((p_1-\lambda)_+\).
Substituting into the definition of \(\Psi_\lambda\) and using
\((p_1-\lambda)_+=p_1-\lambda+(\lambda-p_1)_+\), we obtain
\[
\begin{aligned}
 \Psi_\lambda(L,t')
 =(p_1-\lambda)_++(p_2-\lambda)+r_j(t'-\lambda)
 =p_1+p_2+r_jt'-(2+r_j)\lambda+(\lambda-p_1)_+.
\end{aligned}
\]
Conservation implies \(p_1+p_2+r_jt'=a+b+r_jt\), which is identical in both orders.
The common terms therefore cancel.
Substituting \(p_1=p_\uparrow\) in ascending order and \(p_1=p_\downarrow\) in descending order yields $\Delta_\lambda \ge(\lambda-p_\uparrow)_+-(\lambda-p_\downarrow)_+\ge0,$ because \(p_\uparrow\le p_\downarrow\).
\item For \(\lambda>t_\downarrow\), if \(t_\downarrow<T\), the larger descending supply cylinder residual equals \(t_\downarrow\).
Both residual supply cylinder pressures and the demand cylinder pressure are therefore below \(\lambda\), implying
\(\Psi_\lambda(L_\downarrow,t_\downarrow)=0\).
Hence \(\Delta_\lambda\ge\Psi_\lambda(L_\uparrow,t_\uparrow)\ge0\).
Otherwise, \(t_\uparrow=t_\downarrow=T\), and both orders raise the demand cylinder pressure from \(t\) to \(T\).
The two supply cylinders therefore undergo the same total pressure decrease \(w_{\mathrm{rem}}\), and both pair residual sums equal \(a+b-w_{\mathrm{rem}}\).
Together with \(p_\uparrow\le p_\downarrow\), this proves
\(L_\uparrow\succeq_{\mathrm{maj}}L_\downarrow\).
Because the demand cylinder terms are identical, \eqref{eq:maj-phi} implies
\(\Delta_\lambda\ge\Phi_\lambda(L_\uparrow)-\Phi_\lambda(L_\downarrow)\ge0\).
\end{itemize}
We have proved \(\Delta_\lambda\ge0\) for every \(\lambda\).
The equal final pressure sums therefore imply majorization by~\eqref{eq:maj-phi}.
Repeated adjacent exchanges remove all inversions while preserving completion.
Transitivity of majorization proves the theorem.
\end{proof}

%%%%%%%%%%%%%%%%%%%%%%%%%%%%%%%%%%%%%%%%%%%%%%%%%%%%%%%%%%%%%%%%%%%%
\subsection{Preservation of Majorization under Cascade Filling}
\label{sec:cascade-propagation}

\begin{theorem}
\label{thm:propagation}
For equal-capacity supply cylinders, let \(P=(p_1^{(P)},\ldots,p_m^{(P)})\) and \(Q=(p_1^{(Q)},\ldots,p_m^{(Q)})\) be sorted supply cylinder pressure states satisfying \(P\succeq_{\mathrm{maj}}Q\).
If cascade filling completes demand cylinder \(j\) from \(Q\), then it also completes \(j\) from \(P\), and $F_j(P)\succeq_{\mathrm{maj}}F_j(Q)$.
\end{theorem}
\begin{proof}
For a sorted supply cylinder state $X=(x_1,\ldots,x_m)$, let $S_h=\sum_{i=1}^h x_i$ be the initial pressure sum of the first $h$ cylinders and $S'_h$ their residual pressure sum after the cascade scan, with unused cylinders unchanged.
These residual sums are defined whether the scan reaches the target or fails after scanning all supply cylinders.
Let $w_j:=r_j(\widetilde p_j^{(D)}-\bar p_j^{(D)})$ be the total decrease in supply cylinder pressures required to complete demand cylinder $j$.

Before cylinder $h$, conservation implies that the demand cylinder pressure is $p=\bar p_j^{(D)}+(S_{h-1}-S'_{h-1})/r_j$.
Table~\ref{tab:cascade-prefix-cases} lists the three cases, using $x_h=S_h-S_{h-1}$ and $\alpha_j/r_j=\beta_j=1-\alpha_j$.
\begin{table}[!htb]
\centering\small
\caption{Residual prefix sums in a cascade scan.}
\label{tab:cascade-prefix-cases}
\begin{tabularx}{\linewidth}{@{}>{\raggedright\arraybackslash}p{.23\linewidth}>{\raggedright\arraybackslash}Xl@{}}
\toprule
Case & Pressure behavior & Residual prefix sum $S'_h$\\
\midrule
Cylinder $h$ skipped before completion
& Sorted order implies all earlier cylinders are also skipped.
& $S_h$\\
Equilibration below target
& Both pressures reach $\alpha_jp+\beta_jx_h$.
& $\begin{aligned}[t]
&S'_{h-1}+\alpha_jp+\beta_jx_h\\
&\quad=\alpha_jS'_{h-1}+\beta_jS_h+\alpha_j\bar p_j^{(D)}
\end{aligned}$\\
Target reached in first $h$ cylinders
& Total pressure decrease is $w_j$.
& $S_h-w_j$\\
\bottomrule
\end{tabularx}
\end{table}
Thus, with $S_0=S'_0=0$,
\[
 S'_h=\max\!\left\{S_h-w_j,\,
 \min\!\left\{S_h,\alpha_jS'_{h-1}+\beta_jS_h+\alpha_j\bar p_j^{(D)}\right\}\right\}.
\]
Because $\alpha_j,\beta_j\ge0$, induction shows that $S'_h$ is non-decreasing in each initial prefix sum $S_1,\ldots,S_h$.
Hence $P\succeq_{\mathrm{maj}}Q$ implies $S'_h(P)\le S'_h(Q)$ for every $h$.

\noindent
Let $M=S_m(P)=S_m(Q)$.
Because the scan never raises the demand cylinder pressure above its target, it removes at most $w_j$ from the total supply cylinder pressure.
Completion from $Q$ and the prefix inequality therefore imply $M-w_j\le S'_m(P)\le S'_m(Q)=M-w_j$.
Thus the scan from $P$ removes exactly $w_j$ and completes demand cylinder $j$ by conservation.
Sortedness, equal totals, and the prefix inequalities imply $F_j(P)\succeq_{\mathrm{maj}}F_j(Q)$.
\end{proof}

\begin{theorem}
\label{thm:fixed-order-standardization}
For equal-capacity supply cylinders, if a fixed demand cylinder order $(j_1,\ldots,j_n)$ admits a feasible schedule, then filling each demand cylinder by cascade is also feasible.
After each demand cylinder, the cascade residual state majorizes the residual state of the original schedule.
\end{theorem}
\begin{proof}
Relabel the demand cylinders as \(1,\ldots,n\) in the fixed order and fix a feasible schedule.
Let \(Q_t\) and \(P_t\) be the sorted residual supply cylinder pressure vectors after the first \(t\) fills in the original and cascade schedules, respectively.
Initially, \(P_0=Q_0\).
Suppose the cascade schedule has completed the first \(t-1\) demand cylinders and \(P_{t-1}\succeq_{\mathrm{maj}}Q_{t-1}\).
By Theorem~\ref{thm:single-demand-dominance}, cascading demand cylinder \(t\) from \(Q_{t-1}\) succeeds and leaves a state majorizing \(Q_t\).
Theorem~\ref{thm:propagation} implies success from \(P_{t-1}\) and a residual state majorizing that from \(Q_{t-1}\).
Thus $P_t=F_t(P_{t-1}) \succeq_{\mathrm{maj}}F_t(Q_{t-1}) \succeq_{\mathrm{maj}}Q_t.$
Induction proves feasibility of the cascade schedule and the claimed majorization after every demand cylinder.
\end{proof}

%%%%%%%%%%%%%%%%%%%%%%%%%%%%%%%%%%%%%%%%%%%%%%%%%%%%%%%%%%%%%%%%%%%%
\section{Exact Algorithms and the Lower Bound for Three Supply Cylinders}
\label{sec:frontier-dp}
Throughout this section, supply cylinders have equal capacity. Remove demand cylinders already at their targets and let $D$ be the remaining set, with $n=|D|$.

%%%%%%%%%%%%%%%%%%%%%%%%%%%%%%%%%%%%%%%%%%%%%%%%%%%%%%%%%%%%%%%%%%%%
\subsection{General Frontier Dynamic Programming}
\label{sec:general-frontier-dp}

We use subset dynamic programming for sequencing, following Held and Karp~\cite{heldkarp1962}.
Initially and after each fill, let \(A\subseteq D\) contain the completed demand cylinders and let \(P=(p_1^{(S)},\ldots,p_m^{(S)})\) be the residual supply cylinder state sorted by non-decreasing pressure.
Demand cylinders in $A$ are at their targets and all others at their initial pressures; hence $A$ determines all demand cylinder pressures.
Define the normalized requirement $w_j:=v_j^{(D)}(\widetilde p_j^{(D)}-\bar p_j^{(D)})/v^{(S)}$.
By conservation, completing demand cylinder $j$ reduces the supply cylinder pressure sum by $w_j$; hence every reachable state $(A,P)$ satisfies
\begin{equation}
 \sum_{i=1}^m p_i^{(S)}
 =M(A):=\sum_{i\in\mathcal S}\bar p_i^{(S)}-\sum_{j\in A}w_j.
 \label{eq:residual-sum}
\end{equation}
By~\eqref{eq:residual-sum}, all residual states for the same $A$ have equal sums.
By Theorem~\ref{thm:fixed-order-standardization}, we restrict to cascades. Let $P_0$ be the sorted initial supply cylinder state and $\mathcal R[A]$ the residual states from successful cascade orders of exactly the demand cylinders in $A$.
For the same $A$, Theorem~\ref{thm:propagation} ensures that a majorizing state preserves every feasible remaining cascade demand cylinder order.
As in dominance-based labeling for resource-constrained shortest paths~\cite{pugliese2013,feillet2004}, we discard majorized states and retain the majorization-maximal states as the \emph{frontier}, which may contain several mutually incomparable states.
For a set $\mathcal X$ of states, let $\operatorname{ND}(\mathcal X)$ denote its majorization-maximal states.
Set $\mathcal R[\varnothing]=\mathcal F[\varnothing]=\{P_0\}$ and $\mathcal F[A]:=\operatorname{ND}(\mathcal R[A])$.

\begin{theorem}
\label{thm:frontier-dp}
For THCP with equal-capacity supply cylinders, compute the frontier for each nonempty $A\subseteq D$ as follows.
For each last demand cylinder $j\in A$, cascade from every state in $\mathcal F[A\setminus\{j\}]$, discard failures, and retain the majorization-maximal candidates:
\begin{equation}
 \mathcal C[A]:=\bigcup_{j\in A}\{F_j(P):P\in\mathcal F[A\setminus\{j\}],\ F_j(P)\ne\bot\},\quad
 \mathcal F[A]=\operatorname{ND}(\mathcal C[A]).
 \label{eq:frontier-recurrence}
\end{equation}
Evaluating the recurrence~\eqref{eq:frontier-recurrence} in order of increasing cardinality of $A$ yields the exact frontiers.
Moreover, $\mathcal F[D]\ne\varnothing$ if and only if the THCP instance is feasible.
Predecessors of retained states reconstruct a feasible demand cylinder order.
\end{theorem}

\begin{proof}
Permutations of $A$ induce all states in $\mathcal R[A]$, making it finite. Candidates arise from demand cylinder orders; hence $\mathcal C[A]\subseteq\mathcal R[A]$.
Choose an order reaching \(Q\in\mathcal R[A]\), with last demand cylinder \(j\) and preceding state \(Q'\).
Let $P\in\mathcal F[A\setminus\{j\}]$ majorize $Q'$; such a state exists by finiteness.
Theorem~\ref{thm:propagation} then implies $F_j(P)\in\mathcal C[A]$ and $F_j(P)\succeq_{\mathrm{maj}}Q$.
Every majorization-maximal state in $\mathcal R[A]$ therefore belongs to $\mathcal C[A]$, and every other candidate is majorized by such a state.
Hence $\operatorname{ND}(\mathcal C[A])=\operatorname{ND}(\mathcal R[A])$.
Theorem~\ref{thm:fixed-order-standardization} establishes decision equivalence; predecessors reconstruct a feasible demand cylinder order.
\end{proof}

%%%%%%%%%%%%%%%%%%%%%%%%%%%%%%%%%%%%%%%%%%%%%%%%%%%%%%%%%%%%%%%%%%%%
Let $W:=\max_{A\subseteq D}|\mathcal F[A]|$ be the maximum number of retained states for any completed demand cylinder set.

\begin{theorem}
\label{thm:frontier-complexity}
For THCP with equal-capacity supply cylinders, generate candidates via~\eqref{eq:frontier-recurrence} and retain the majorization-maximal candidates by all-pairs comparison.
This computes all frontiers using $O(m\log(m+1)+mn^2 2^nW^2)$ rational arithmetic operations and comparisons and stores $O(m2^nW)$ rational values and indices.
For binary rational inputs, the bit complexity is $2^nW^2\operatorname{poly}(L)$ time and $2^nW\operatorname{poly}(L)$ space.
\end{theorem}
\begin{proof}
Each $A$ has at most $|A|W$ candidates. Cascades preserve sorted order (Lemma~\ref{lem:cascade-supply-order}), making each fill and comparison $O(m)$ and all-pairs frontier extraction $O(m|A|^2W^2)$ operations.
Summing over all subsets using $\sum_{A\subseteq D}|A|^2=n(n+1)2^{n-2}$ proves the time bound, including initial sorting.
Storing at most $2^nW$ frontier states and $nW$ temporary candidates proves the space bound.
Each stored state is reached by a legal schedule of at most $mn$ transfers.
Its pressures therefore have polynomial bit length by Lemma~\ref{lem:exact-simulation}.
Prefix sums and cross-products in feasibility tests and frontier comparisons retain polynomial bit length.
Predecessor indices require $O(\log(nW+1))=O(n\log(n+1))$ bits because $W\le n!$.
\end{proof}

%%%%%%%%%%%%%%%%%%%%%%%%%%%%%%%%%%%%%%%%%%%%%%%%%%%%%%%%%%%%%%%%%%%%
\subsection{Dynamic Programming for Two Supply Cylinders}
\label{sec:two-supply-representative-dp}
For a completed set $A$, conservation fixes the residual total at $M(A)$.
For $0\le x\le M(A)/2$, define $X_A(x):=(x,M(A)-x)$.
For $x,y$ in this interval, $X_A(x)\succeq_{\mathrm{maj}}X_A(y)$ exactly when $x\le y$; hence each nonempty frontier has one state minimizing the lower pressure, and $W=1$.
Let $\ell[A]$ be this minimum, or $\bot$ if the frontier is empty.

For $j\notin A$, let $c_j(A,x)$ be the lower component of $F_j(X_A(x))$, or $\bot$ on failure; two supply updates suffice.
Initialize $\ell[\varnothing]=\min\{\bar p_1^{(S)},\bar p_2^{(S)}\}$ and evaluate, in increasing set size,
\[
 \ell[A]=\min_{\substack{j\in A\\\ell[A\setminus\{j\}]\ne\bot}}
 c_j\bigl(A\setminus\{j\},\ell[A\setminus\{j\}]\bigr).
\]
Discard $\bot$ values and set $\min\varnothing:=\bot$.

\begin{theorem}
\label{thm:two-supply-dp}
For THCP with exactly two equal-capacity supply cylinders and arbitrary positive rational demand cylinder capacities, feasibility can be decided using $O(n2^n)$ rational operations and comparisons.
The algorithm stores $O(2^n)$ rational values and indices.
Its bit complexity is $2^n\operatorname{poly}(L)$ in both time and space.
A feasible schedule can be reconstructed within the same bounds.
\end{theorem}
\begin{proof}
Theorem~\ref{thm:frontier-dp} and the total order above prove the recurrence and feasibility criterion $\ell[D]\ne\bot$.
Precompute $M(A)=M(A\setminus\{j\})-w_j$ for any $j\in A$, starting from $M(\varnothing)$, in $O(2^n)$ operations.
Each candidate takes $O(1)$ operations; hence $\sum_A|A|=n2^{n-1}$ yields the time bound.
Store $M(A)$, $\ell[A]$, and an attaining last demand per subset.
Following these predecessors and replaying the resulting cascade order reconstructs a schedule in $O(n)$ operations.
The bit bounds follow from Theorem~\ref{thm:frontier-complexity}.
\end{proof}
Corollary~\ref{cor:target-unit-eth} shows that the running-time exponent in Theorem~\ref{thm:two-supply-dp} is optimal up to a constant factor under ETH, even when all capacities are one.

%%%%%%%%%%%%%%%%%%%%%%%%%%%%%%%%%%%%%%%%%%%%%%%%%%%%%%%%%%%%%%%%%%%%
\subsection{An ETH Lower Bound with Three Supply Cylinders}
\label{sec:three-supply-lower-bound}

By Theorem~\ref{thm:fixed-order-standardization}, enumerating all demand cylinder orders and applying cascade filling decides THCP in $n!\operatorname{poly}(L)=2^{O(n\log n)}\operatorname{poly}(L)$ time.
We prove that ETH rules out a $2^{o(n\log n)}\operatorname{poly}(L)$-time algorithm even with three supply cylinders, by reducing sparse 3SAT to THCP via the following arithmetic problem.

\smallskip\noindent\textsc{Exact Permutation Sum} (EPS).
The input consists of row weights $U_1<\cdots<U_k$ and label values $V_1,\ldots,V_k$, all positive integers, and an integer target $T_{\rm src}$.
The question is whether some $\pi\in S_k$ satisfies $\sum_{i=1}^k U_iV_{\pi(i)}=T_{\rm src}$.
A permutation assigns one label to each row, using each label exactly once.
Assigning label $j$ to row $i$ contributes $U_iV_j$ to the sum.

First, Lemma~\ref{lem:eth-numerical-source} compresses a sparse 3SAT formula with $N$ variables into at most $2^{o(N)}$ EPS instances, each with $k=\Theta(N/\log N)$ rows.
The formula is satisfiable if and only if at least one of these instances is a yes-instance.
Then, Lemma~\ref{lem:eth-physical-realization} transforms each EPS instance into a THCP instance with three supply cylinders and $n=3k+1$ demand cylinders.
Because $n\log n=\Theta(N)$, a $2^{o(n\log n)}\operatorname{poly}(L)$-time algorithm for THCP would solve sparse 3SAT in $2^{o(N)}$ time, contradicting ETH.

\begin{lemma}
\label{lem:eth-numerical-source}
For each fixed $c_{\rm sat}>0$, a deterministic algorithm maps a 3CNF formula $F$ with $N$ variables and at most $c_{\rm sat}N$ clauses to at most $2^{o(N)}$ EPS instances in $2^{o(N)}\operatorname{poly}(|F|)$ time.
Each instance has $k=\Theta(N/\log N)$ rows, uniformly in $F$, and each of its integers has $O(N)$ bits.
The formula is satisfiable if and only if some instance is a yes-instance.
Every solution permutation decodes to a satisfying assignment in polynomial time.
\end{lemma}

\begin{proof}
For sufficiently large $N$, partition variables into blocks of at most $r=\lfloor\tfrac12\log_2N\rfloor$ variables and clauses into blocks of at most $s=\lfloor r/3\rfloor$ clauses.
A variable pattern lists its block's variable values, and a clause pattern lists literal values satisfying every clause in its block.
There are $\kmain=\Theta(N/\log N)$ blocks.
A choice of one pattern per block is consistent when each variable bit equals its corresponding literal bit, complemented for a negative occurrence.
Appendix~\ref{app:eth-palettes} proves that a consistent choice exists if and only if $F$ is satisfiable.
Group literal occurrences by $h=(t,u,\pm)$, where $t$ is the variable's position in its variable block, $u$ is the literal's position in its clause block, and the sign records polarity; this fixes the bit positions read from each pattern.
In a group of $d_h$ occurrences, weights $1,2,4,\ldots$ encode the two compared bit sequences as binary integers.
The integers differ by at most $2^{d_h}-1$ and are equal exactly when all comparisons agree.

First, combine the tests into one product-sum congruence.
For example, let $h=(1,u,+)$ consist of two occurrences of one variable in distinct clause blocks $A,B$.
Let $x_1$ be the variable bit and $y_A,y_B$ the corresponding bits in $A,B$.
Weights $1,2$ yield 
\[
(x_1-y_A)+2(x_1-y_B)=3x_1-y_A-2y_B.
\]
The block coefficients $3,-1,-2$ multiply the selected bits.

For general groups, choose pairwise coprime prime powers $q_h>2^{d_h}-1$; define $\Pmod=\prod_hq_h$, $e_h=\Pmod/q_h$.
The example contributes $3e_h,-e_h,-2e_h$ to the variable block and clause blocks $A,B$.
Add each occurrence's binary weight times $e_h$ to its variable block and subtract it from its clause block, summing over all groups and all variables or literals in each block.
Reduce totals to least nonnegative residues modulo $\Pmod$, then add $i\Pmod$ for block $i$ to obtain positive, distinct $a_i$ without changing residues.
Independently of placement, sum $e_h$ once per nonempty group reading a $1$ from pattern $p$, complementing negative clause bits; its least nonnegative residue modulo $\Pmod$ is $b_p$.
Let $\tau(i)$ be the matching-type pattern selected by block $i$, and let $S=\sum_i a_i b_{\tau(i)}$.
Modulo $q_h$, each factor retains only its $e_h$-term, yielding $e_h^2$ times the group's binary difference.
The example yields $e_h^2(3x_1-y_A-2y_B)$.
Because $e_h^2$ is invertible modulo $q_h$ and the difference has magnitude below $q_h$, the residue vanishes exactly when that group's comparisons agree.
Pairwise coprimality yields $S\equiv0\pmod{\Pmod}$ exactly when $\tau$ is consistent.
Appendix~\ref{app:eth-occurrence-crt} proves this equivalence and $\log\Pmod=O(N)$.

Next, fix the EPS target before choosing a permutation, although consistent main sums $S$ vary.
From the input, fix a common $\Pmod$-multiple $T_{\rm raw}$ above all main sums and fix auxiliary row weights as $\Pmod$ times powers of $\kmain$.
Consistency makes $(T_{\rm raw}-S)/\Pmod$ an integer; its base-$\kmain$ digits specify auxiliary labels filling the gap.
Conversely, auxiliary products vanish modulo $\Pmod$; every matching-type exact solution therefore has a consistent main choice.
Enumerate copy counts of main patterns and auxiliary digits, leaving placement to the permutation.
Appendix~\ref{app:eth-exact-lifting} constructs $O(\kmain)$ auxiliary rows and proves both directions.
With $O(\kmain)$ permitted symbols and $O(\kmain)$ label copies, there are at most $2^{O(\kmain)}=2^{o(N)}$ inventories, each fixing the label multiset $(b_j)$ before its permutation is chosen.

Finally, require matching row and label types.
Let row and label tags $f_i,g_j$ be $1,2,3$ for variable, clause, and auxiliary types, respectively.
Let $H=1+\sum_i f_i^2$, and replace raw values by $Ha_i+f_i$ and $Hb_j+g_j$.
Define the final target $T_{\rm src}$ from the inventory, independently of the permutation, with remainder $H-1$ modulo $H$.
Each product satisfies $(Ha_i+f_i)(Hb_j+g_j)\equiv f_i g_j\pmod H$.
Row and label type counts agree.
Thus $0\le\sum_i f_i g_{\sigma(i)}\le H-1$ for every permutation $\sigma$, with the maximum attained exactly when all paired types match.
The final equation forces matching types and is equivalent to $\sum_i a_i b_{\sigma(i)}=T_{\rm raw}$ on those placements.
Appendix~\ref{app:eth-type-enforcement} proves these claims and verifies the EPS input conditions and polynomial-time decoding.
Together with Appendix~\ref{app:eth-exact-lifting}, this yields $k=\Theta(N/\log N)$ and $O(N)$ bits per integer.
A satisfying assignment determines an output EPS solution; every output solution decodes to a satisfying assignment of $F$.
\end{proof}

We reduce EPS to THCP using three unit-capacity supplies and $3k+1$ demands.
Let $B_j,A_i,T_i,G$ denote label, slot, clock, and probe demand cylinders, respectively.
The nonprobe demands are called \emph{source demands}.
A \emph{word} lists demand cylinders in their cascade filling order.
For $\pi\in S_k$, define the \emph{canonical word} $w_\pi=B_{\pi(1)}A_1T_1\cdots B_{\pi(k)}A_kT_kG$.
The consecutive triple $B_{\pi(i)}A_iT_i$ forms \emph{block $i$} of $w_\pi$.
Block $i$ assigns label $\pi(i)$ to row $i$.

Let $\varepsilon>0$ be a small encoding parameter.
The \emph{baseline} sets $\varepsilon=0$ while retaining all capacities; its $k$ labels then have identical parameters, with initial pressure $1$ and target pressure $3$, and their order does not affect the supply trajectory.
Appendix~\ref{app:eth-encoding-strategy} specifies the capacities, pressure targets, and encoding parameters.
Table~\ref{tab:eth-demand-types} summarizes the demands in normalized pressures.

\begin{table}[!htb]
\centering\small
\caption{Demand counts, initial pressures, and roles in canonical words.}
\label{tab:eth-demand-types}
\setlength{\tabcolsep}{4.5pt}
\renewcommand{\arraystretch}{1.08}
\begin{tabularx}{\linewidth}{@{}lcc>{\raggedright\arraybackslash}X@{}}
\toprule
Demand & Count & Initial pressure & Role\\
\midrule
Label $B_j$ & $k$ & $1-\varepsilon V_j$ & Encodes $V_j$; label order selects $\pi$.\\
Slot $A_i$ & $k$ & $1$ & Represents row $i$.\\
Clock $T_i$ & $k$ & $0$ & Separates blocks; $T_k$ tests $\sum_i U_iV_{\pi(i)}\le T_{\rm src}$.\\
Probe $G$ & $1$ & $1$ & Last demand; tests $\sum_i U_iV_{\pi(i)}\ge T_{\rm src}$.\\
\bottomrule
\end{tabularx}
\end{table}

\begin{lemma}[Canonical words succeed exactly for EPS solutions]
\label{lem:eth-canonical-computation}
For every $\pi\in S_k$, the canonical word $w_\pi$ succeeds exactly when $\sum_i U_iV_{\pi(i)}=T_{\rm src}$.
\end{lemma}

\begin{proof}For canonical words, label order encodes the EPS sum, and the final clock and probe impose opposite inequalities.
Appendix~\ref{app:eth-canonical-feasibility} proves that every canonical word proceeds until $T_k$ equilibrates with supply 2, whether or not the EPS equality holds.
Fix $\pi$ and, immediately after this equilibration, let $p_1,p_2$ be the actual pressures of supplies 1 and 2, and $p_1^0$ the pressure of supply 1 in the baseline.
The subsequent contact between $T_k$ and supply 3 changes neither $p_1$ nor $p_2$.
Let $d>0$ be the small scale chosen for the fixed EPS input.
Label $B_{\pi(i)}$ starts $\varepsilon V_{\pi(i)}$ below its baseline pressure, inducing a pressure change in supply 1.
Each block scales this change; the capacities make the combined multiplier from block $i$ onward proportional to $1+dU_i$.
Let $S=\sum_jV_j$.
For a common coefficient $\gamma>0$ independent of $\pi$, summing the label contributions yields
\[
 p_1-p_1^0
 =-\gamma\sum_i(1+dU_i)V_{\pi(i)}
 =-\gamma\left(S+d\sum_iU_iV_{\pi(i)}\right).
\]
Because $S$ is independent of the order, only the EPS sum varies with $\pi$.
Appendix~\ref{app:eth-baseline} defines the reference state $C=(C_1,C_2,C_3)$ from the EPS input, without assuming reachability.
If the EPS equality holds, $C$ is the supply state before the probe, and the final clock equilibrates with supply 3 at its target.
Replacing the EPS sum in the pressure formula by $T_{\rm src}$ yields $C_1=p_1^0-\gamma(S+dT_{\rm src})$.
Subtracting this reference pressure yields
\[
 x_\pi:=p_1-C_1
 =\alpha\left(T_{\rm src}-\sum_iU_iV_{\pi(i)}\right),
 \qquad \alpha=\gamma d>0.
\]
Thus the EPS equality holds exactly when $x_\pi=0$.
By the construction in Appendix~\ref{app:eth-encoding-strategy}, the checkpoint state satisfies
\[
 3p_1+2p_2=3C_1+2C_2,
 \qquad p_2-C_2=-\frac32x_\pi.
\]
The second relation follows by substituting $p_1-C_1=x_\pi$.

The final clock \(T_k\) has target pressure \(C_3\).
Appendix~\ref{app:eth-baseline} calibrates $C$ to match the actual total gas of the supplies and unfinished clock, with supply 3 and the clock both at $C_3$ in the reference.
Because supplies 1 and 2 have unit capacity, their gas deviation from the reference is $x_\pi-3x_\pi/2=-x_\pi/2$.
The gas in supply 3 and the unfinished clock therefore differs by $x_\pi/2$ from the amount needed to bring both to $C_3$.
Thus $T_k$ completes exactly when $x_\pi\ge0$.
On completion, transfer stops at $C_3$, leaving the excess gas $x_\pi/2$ in unit-capacity supply 3.
The resulting supply state is $(C_1+x_\pi,C_2-3x_\pi/2,C_3+x_\pi/2)$.

If $T_k$ completes, we execute the probe $G$.
Its first two contacts equilibrate below target by Appendix~\ref{app:eth-common-routes}, leaving completion to the third contact.
For the capacity-$9$ probe, the resulting equilibrium arrival pressure minus its target is $-x_\pi/250$.
Thus $G$ completes exactly when $x_\pi\le0$.
Together with the clock condition $x_\pi\ge0$, this makes $w_\pi$ successful exactly when $x_\pi=0$, equivalently when $\sum_iU_iV_{\pi(i)}=T_{\rm src}$.
\end{proof}

\begin{lemma}[Successful words are canonical]
\label{lem:eth-encoded-normal-form}
Every successful full word consists of the blocks $B_{\pi(i)}A_iT_i$ in index order, followed by $G$, for some $\pi\in S_k$; only the label order can vary.
\end{lemma}

\begin{proof}
Fix a successful full word; Appendix~\ref{app:eth-physical-order} shows that $G$ is last and the clocks occur in index order.
Let $p=(p_1,p_2,p_3)$ be the supply state immediately before $G$, and use the reference state $C$ from Lemma~\ref{lem:eth-canonical-computation}.
Let $E(p)=3p_1+2p_2$, and let $E_0$ be the algebraic baseline endpoint potential for the same source order, defined even if baseline filling fails; its canonical value is $E(C)$.

Success requires $E(p)\le E(C)$, whereas a noncanonical source order will yield $E(p)=E_0+O(d^4)>E(C)$ at the chosen scale.
Because $C$ is calibrated to the total supply gas after completing these same source demands, unit supply capacities yield $p_1+p_2+p_3=C_1+C_2+C_3$.
Completion of $T_k$ leaves supply 3 at pressure at least $C_3$.
Later labels and slots leave it unchanged; hence $p_3\ge C_3$.
The probe's first two contacts equilibrate below target.
These conditions imply $E(p)\le E(C)$ by Appendix~\ref{app:eth-shared-terminal}.

We first show that deleting all labels leaves $A_1T_1\cdots A_kT_k$.
Each slot $A_j$ and clock $T_j$ has a positive construction weight $\omega_j$.
In the canonical order, both $A_j$ and $T_j$ precede $A_i$ for each $j<i$, contributing total weight $2\sum_{j<i}\omega_j$.
Let $\Delta_i$ be the actual preceding slot--clock weight minus this canonical value.
This mismatch affects the supply-1 pressure when the slot is processed.
After summing over slots, the corrections caused by switching to supply 2 and the remaining terms combine through $2(\Delta_i)_+-\Delta_i=|\Delta_i|$.
For the construction scale $d>0$, Appendix~\ref{app:eth-endpoint-certificates} proves 
\[
E_0-E(C)=2d^2\sum_i\omega_i|\Delta_i|+O(d^3).
\]
Because all weights are positive, the coefficient of $d^2$ vanishes exactly when every slot has its canonical preceding weight.
The canonical preceding totals increase strictly with the slot index, forcing slots into index order.
Before $A_i$, the preceding slots then contribute $\sum_{j<i}\omega_j$, leaving the same required weight for clocks.
Because clocks are ordered and their weights are positive, exactly $T_1,\ldots,T_{i-1}$ precede $A_i$.
Thus the coefficient vanishes exactly for $A_1T_1\cdots A_kT_k$ and is positive otherwise.
For this slot--clock order, we next show that one label must immediately precede each pair $A_iT_i$.
Let $N_i$ count labels before $A_i$ and $\ell_i$ those between $A_i$ and $T_i$.
Label-placement errors appear at the next order: 
\[
E_0-E(C)=d^3\sum_i\omega_i\bigl(|N_i-i|+\ell_i\bigr)+O(d^4).
\]
Because all weights are positive, the nonnegative coefficient vanishes exactly when $N_i=i$ and $\ell_i=0$ for every $i$.
These conditions place one label immediately before each pair, with none inside or after; every other placement has a positive $d^3$ coefficient.

Finally, we verify that label encoding preserves the positive leading term for every noncanonical source order with clocks in index order.
With $\varepsilon=d^2$, label pressure shifts and capacities are both $O(d^2)$, producing $O(d^4)$ supply-state errors that subsequent updates do not amplify.
Appendix~\ref{app:eth-prefix-estimate} bounds the endpoint error uniformly over source orders for the fixed EPS input, leaving the positive $d^2$ or $d^3$ leading term unchanged.
Appendix~\ref{app:eth-stability} selects one polynomial-bit rational $d$ preserving these positive signs.
A noncanonical successful word would therefore satisfy $E(p)>E(C)$, contradicting the terminal bound.
\end{proof}

\begin{lemma}
\label{lem:eth-physical-realization}
Every EPS instance reduces in polynomial time to an equivalent THCP instance with three unit-capacity supply cylinders, $3k+1$ demand cylinders, positive integer pressures, and positive rational demand cylinder capacities.
Its bit length is polynomial.
\end{lemma}

\begin{proof}
An EPS solution yields a successful canonical word by Lemma~\ref{lem:eth-canonical-computation}.
Conversely, Theorem~\ref{thm:fixed-order-standardization} converts any feasible schedule into a successful cascade word.
Lemma~\ref{lem:eth-encoded-normal-form} makes that word canonical, and Lemma~\ref{lem:eth-canonical-computation} recovers the EPS equality.
Appendix~\ref{app:eth-output} verifies polynomial construction time and all output conditions, including bit length.
\end{proof}

\begin{theorem}
\label{thm:three-supply-eth}
Assume ETH. Then THCP with exactly three unit-capacity supply cylinders admits no deterministic $2^{o(n\log n)}\operatorname{poly}(L)$-time algorithm, even when all initial and target pressures are positive integers and demand cylinder capacities are arbitrary positive rationals.
\end{theorem}

\begin{proof}
Under ETH, there is a constant $\delta>0$ for which 3SAT on $N$ variables cannot be decided in $O(2^{\delta N})\operatorname{poly}(|F|)$ time.
Let $\rho=\delta/2$.
Deterministic sparsification~\cite[Corollary~1]{ipz2001} maps the input to at most $2^{\rho N}$ formulas of sparsity $c_{\rm sat}(\rho)$.
Lemma~\ref{lem:eth-numerical-source} produces an OR of $2^{o(N)}$ EPS instances per sparse formula, each with $O(kN)$ input bits, in $2^{o(N)}\operatorname{poly}(|F|)$ time.
Apply Lemma~\ref{lem:eth-physical-realization} to every resulting EPS instance.
Each EPS instance has $k=\Theta(N/\log N)$ rows and produces $n=3k+1$ demand cylinders with input length $L=N^{O(1)}$.

Thus $n\log n=\Theta(N)$ and $\operatorname{poly}(L)=2^{o(N)}$, and the resulting OR is equivalent to the input formula.
A $2^{o(n\log n)}\operatorname{poly}(L)$ algorithm would decide 3SAT in $2^{\rho N+o(N)}\operatorname{poly}(|F|)$ time, including generation.
All reduction constants are fixed by the choice of $\rho$.
For sufficiently large $N$, $\rho N+o(N)\le\delta N$, contradicting the choice of $\delta$.
\end{proof}

%%%%%%%%%%%%%%%%%%%%%%%%%%%%%%%%%%%%%%%%%%%%%%%%%%%%%%%%%%%%%%%%%%%%

\section{Conclusion}
Our results establish the usefulness and limits of majorization for pressure-driven charging.
For equal-capacity supply cylinders, conservation fixes each completed set's residual pressure total, ruling out componentwise dominance between distinct states.
Theorems~\ref{thm:propagation}--\ref{thm:frontier-dp} justify restricting to cascades and discarding majorized states in subset dynamic programming.
With two supply cylinders, majorization totally orders each completed set's states, and Theorem~\ref{thm:two-supply-dp} establishes a \(2^n\operatorname{poly}(L)\)-time algorithm using one representative per set.

Despite preserved majorization, Theorem~\ref{thm:three-supply-eth} rules out deterministic \(2^{o(n\log n)}\operatorname{poly}(L)\)-time algorithms for three supply cylinders under ETH.
With demand cylinder order enumeration and Corollary~\ref{cor:target-unit-eth}, these bounds establish optimal deterministic exponents \(\Theta(n)\) and \(\Theta(n\log n)\) for two and three equal-capacity supply cylinders under ETH, allowing arbitrary positive rational demand cylinder capacities.
Thus, constant-dimensional residual pressure states and preserved majorization do not ensure deterministic single-exponential exact algorithms under ETH.
In particular, Theorems~\ref{thm:frontier-complexity} and~\ref{thm:three-supply-eth} imply that the frontier width \(W\) cannot be bounded by \(2^{O(n)}\) in the worst case for three equal-capacity supply cylinders.
An open question is whether unit demand cylinder capacities permit a single-exponential algorithm with three equal-capacity supply cylinders.
Resolving this question would clarify the role of varying demand cylinder capacities in the exponent gap.

\subparagraph*{AI use.}
OpenAI Codex assisted manuscript editing and proof development, including reductions.
The authors verified all arguments and retain full responsibility.

\clearpage
\bibliography{references}
\clearpage
\appendix

%%%%%%%%%%%%%%%%%%%%%%%%%%%%%%%%%%%%%%%%%%%%%%%%%%%%%%%%%%%%%%%%%%%%
\section{From Sparse 3SAT to Exact Permutation Sum}
\label{app:numerical-source}

\subsection{Consistency and Satisfiability}
\label{app:eth-palettes}
We prove that the local patterns in the main proof admit a consistent choice if and only if $F$ is satisfiable.
Output no instances if $F$ has an empty clause; otherwise, pad short clauses to length three by repeating literals.
Handle finitely many small input sizes by direct decision and fixed YES or NO instances.
For all remaining inputs, use the block sizes $r,s$ and partition in the main proof.
Let $\Ncl$ be the clause count, $k_v=\lceil N/r\rceil$ the variable block count, and $k_c=\lceil\Ncl/s\rceil$ the clause block count.
Define $\kmain=k_v+k_c$ and index the main rows by $i=1,\ldots,\kmain$; here $r\ge3$, $s\ge1$, and $\kmain\ge2$.
All unspecified logarithms have base two.

Define the variable and clause pattern sets by $\mathcal V=\{0,1\}^r$ and $\mathcal C=(\{0,1\}^3\setminus\{000\})^s$, respectively.
Patterns retain their variable or clause type even when their bit strings coincide.
Let $\tau$ select a pattern from $\mathcal V$ on each variable row and from $\mathcal C$ on each clause row.
Unused positions impose no consistency constraints; when extending an assignment to full patterns, fill unused variable positions with $0$ and unused clause triples with $111$.

If $\tau$ is consistent in the sense of the main proof, its variable patterns define a global assignment, and its clause patterns record the actual literal values.
Because every clause triple in $\mathcal C$ contains a $1$, this assignment satisfies $F$.
Conversely, every satisfying assignment determines variable patterns in $\mathcal V$ and clause patterns in $\mathcal C$; completing unused positions as above yields a consistent choice $\tau$.

\subsection{Consistency Test and Size Bounds}
\label{app:eth-occurrence-crt}
\label{app:eth-typed-congruence}
We first prove that the construction in the main proof tests consistency exactly.
For each group $h$ of $d_h$ occurrences, fix an ordering $\nu(e)=0,\ldots,d_h-1$ specifying their binary weights.
For a matching-type choice $\tau$, let $S=\sum_{i=1}^{\kmain}a_i b_{\tau(i)}$.
Let $x_e,y_e$ be its variable and literal bits at occurrence $e$, complementing $y_e$ for a negative literal.
Fix a group $h$.
Modulo $q_h$, all other cofactors and the row shifts $i\Pmod$ vanish.
Because $h$ fixes the bit position on each side, expanding the main proof's products yields
\begin{align}
 S&\equiv e_h^2\sum_{e\in h}2^{\nu(e)}(x_e-y_e)\pmod{q_h},
 \\
 S&\equiv0\pmod{\Pmod}
 \quad\Longleftrightarrow\quad\tau\text{ is consistent}.
 \label{eq:eth-modular-source}
\end{align}
To justify the second line, the weighted difference is the difference of two binary integers in $[0,2^{d_h}-1]$.
Its magnitude is below $q_h$; uniqueness of binary representation implies that it is zero exactly when all corresponding bits agree.
Because $\gcd(e_h,q_h)=1$, $e_h^2$ is invertible modulo $q_h$; the first congruence thus has zero remainder exactly when every comparison in $h$ agrees.
Pairwise coprimality then combines these conditions over all groups.
In particular, errors in different groups cannot cancel.
If there are no groups, define $\Pmod=1$, $a_i=i$, and $b_p=0$; consistency is automatic.

We then prove the size and construction-time bounds.
Let $g$ count the nonempty groups.
For $g>0$, take the first $g$ odd primes $p_h$ and choose $q_h$ as the least power of $p_h$ exceeding $2^{d_h}-1$.
We have $g\le6rs\le2r^2=O(\log^2N)$ and $\sum_hd_h=3\Ncl=O(N)$.
The bound $p_h=O((g+1)\log(g+2))$~\cite[Corollary to Theorem~3]{rosserschoenfeld1962} permits polynomial-time sieving.
Because $q_h\le p_h(2^{d_h}-1)$, we have $\log\Pmod\le\sum_hd_h+O(g\log(g+2))=O(N)$.
The main proof's residue choices satisfy $i\Pmod\le a_i<(i+1)\Pmod$ and $0\le b_p<\Pmod$; hence all main values have $O(N)$ bits.
These bounds permit polynomial-time prime-power selection and construction of all occurrence weights, row values, and pattern values.

\subsection{Exact Lifting and Inventory Enumeration}
\label{app:eth-exact-lifting}
\label{app:eth-source-budget}
Let $\kmain=k_v+k_c\ge2$ be the main row count.

The ranges in Appendix~\ref{app:eth-occurrence-crt} imply $0\le S=\sum_{i=1}^{\kmain}a_i b_{\tau(i)}<\kmain(\kmain+1)\Pmod^2$ for every matching-type main choice $\tau$, without assuming consistency.
Let $s_0$ be the least integer such that $\kmain^{s_0}>\kmain(\kmain+1)\Pmod$, found by repeated integer multiplication.
Define one auxiliary row per base-$\kmain$ digit position as in Table~\ref{tab:eth-auxiliary-digits}, with permitted labels $0,\ldots,\kmain-1$.
The common target $T_{\rm raw}=\Pmod(\kmain^{s_0}-1)$ is the maximum auxiliary sum, attained by selecting $\kmain-1$ at every position.
The auxiliary rows, permitted labels, and target are fixed before any main choice.
\begin{table}[!ht]
\centering\small
\caption{Auxiliary rows for base-$\kmain$ completion, with $0\le\nu<s_0$ and $0\le d_\nu<\kmain$.}
\label{tab:eth-auxiliary-digits}
\begin{tabular}{@{}lcccc@{}}
\toprule
Digit position $\nu$ & $0$ & $1$ & $\cdots$ & $s_0-1$\\
\midrule
Row value $a_{\kmain+\nu+1}$ & $\Pmod$ & $\Pmod\kmain$ & $\cdots$ & $\Pmod\kmain^{s_0-1}$\\
Selected label $d_\nu$ & $d_0$ & $d_1$ & $\cdots$ & $d_{s_0-1}$\\
Row contribution & $\Pmod d_0$ & $\Pmod\kmain d_1$ & $\cdots$ & $\Pmod\kmain^{s_0-1}d_{s_0-1}$\\
\bottomrule
\end{tabular}
\end{table}

\subparagraph*{Correctness of the lifting.}
Integrality implies $\kmain^{s_0}-1\ge\kmain(\kmain+1)\Pmod$; hence $T_{\rm raw}\ge\kmain(\kmain+1)\Pmod^2>S$ for every such choice.
For a consistent main choice, equivalence~\eqref{eq:eth-modular-source} implies that $S$ is a multiple of $\Pmod$.
Thus $(T_{\rm raw}-S)/\Pmod$ is an integer in $[0,\kmain^{s_0})$.
Its base-$\kmain$ expansion, padded with leading zeros, selects labels $d_\nu$ in Table~\ref{tab:eth-auxiliary-digits} with $T_{\rm raw}-S=\Pmod\sum_{\nu=0}^{s_0-1}\kmain^\nu d_\nu$.
Thus the auxiliary rows complete the main sum to the fixed target; only the selected digits depend on the main choice.

We prove that a consistent main choice exists if and only if some enumerated candidate label multiset admits a matching-type exact solution.
Each candidate, called an inventory, specifies the number of copies of each permitted symbol, with $k_v$ variable-pattern copies from $\mathcal V$, $k_c$ clause-pattern copies from $\mathcal C$, and $s_0$ auxiliary-digit copies from $0,\ldots,\kmain-1$.
Each inventory fixes $k=\kmain+s_0$ labels before a permutation chooses their positions.
The rows and $T_{\rm raw}$ are common to all inventories.
Let $b_j$ be the numerical value of label copy $j$: its pattern value from Appendix~\ref{app:eth-occurrence-crt} for a main label, or its digit value for an auxiliary label.
Retain the type and original symbol with each copy index, even when numerical values coincide.

A consistent main choice and its completion determine one of these inventories and a matching-type permutation $\sigma$ satisfying
\begin{equation}
 \sum_{i=1}^k a_i b_{\sigma(i)}=T_{\rm raw}.
 \label{eq:eth-typed-exact}
\end{equation}

Conversely, reducing a matching-type solution of~\eqref{eq:eth-typed-exact} modulo $\Pmod$ removes the auxiliary products and target, because each is a multiple of $\Pmod$.
This leaves $S\equiv0\pmod{\Pmod}$, and equivalence~\eqref{eq:eth-modular-source} establishes consistency of the main patterns.

\subparagraph*{Size and enumeration bounds.}
We bound the row count, inventory count, and integer bit length.
\begin{enumerate}[(i)]
\item Row count. Because $\kmain=\Theta(N/\log N)$ and Appendix~\ref{app:eth-occurrence-crt} proves $\log\Pmod=O(N)$, we have $s_0=O(1+\log_{\kmain}\Pmod)=O(\kmain)$ and $k=\Theta(N/\log N)$.
\item Number of inventories. The inequalities $3s\le r$, $N\ge2^{2r}$, and $r\le2^{r-1}$ imply $|\mathcal V|+|\mathcal C|\le2^{r+1}\le N/r\le k_v\le\kmain$.
Including the $\kmain$ auxiliary digits, there are at most $2\kmain$ typed label symbols.
Encode each multiset by one star per copy and one separator between successive permitted symbols.
The combined encoding has fixed length at most $k+|\mathcal V|+|\mathcal C|+\kmain\le s_0+3\kmain$, permitting at most $2^{s_0+3\kmain}=2^{O(\kmain)}=2^{o(N)}$ inventories.
Enumerate the three multisets recursively with polynomial delay.
\item Bit length and construction time. Minimality of $s_0$ implies $\kmain^{s_0}\le\kmain^2(\kmain+1)\Pmod$; hence all raw values and $T_{\rm raw}$ have $O(N)$ bits.
Appendix~\ref{app:eth-type-enforcement} preserves the row and bit bounds with polynomial-time construction per inventory, yielding $2^{o(N)}\operatorname{poly}(|F|)$ total time.
\end{enumerate}

\subsection{Positive Type Tags and Decoding}
\label{app:eth-type-enforcement}
\label{app:eth-strict-rows}
For the transformed rows $A_i=Ha_i+f_i$ and labels $B_j=Hb_j+g_j$ from Lemma~\ref{lem:eth-numerical-source}, define the final target for each inventory by
\[
 T_{\rm src}=H^2T_{\rm raw}
   +H\sum_i f_i a_i+H\sum_j g_j b_j+H-1.
\]
The label sum counts this inventory's copies; the target therefore depends only on copy counts and is fixed before any permutation is chosen.

The inventory construction makes the row and label tag multisets equal.
Thus $H-1-\sum_i f_i g_{\sigma(i)}=\tfrac12\sum_i(f_i-g_{\sigma(i)})^2\ge0$ for any permutation $\sigma$.
Because the tags are positive, the tag score lies in $[0,H-1]$, reaching its maximum exactly when every paired tag agrees.
If $\sum_i A_iB_{\sigma(i)}=T_{\rm src}$, reduction modulo $H$ yields $\sum_i f_i g_{\sigma(i)}\equiv H-1\pmod H$.
Within this range, the congruence forces the maximum score itself, hence matching types.

Matching types imply $\sum_i a_i g_{\sigma(i)}=\sum_i a_i f_i$ and $\sum_i f_i b_{\sigma(i)}=\sum_j g_j b_j$.
Expanding the products and subtracting the chosen target therefore yields
\[
 \sum_i A_iB_{\sigma(i)}-T_{\rm src}
 =H^2\left(\sum_i a_i b_{\sigma(i)}-T_{\rm raw}\right).
\]
Thus a permutation satisfies the final equation if and only if it preserves types and satisfies the raw equation.
Appendix~\ref{app:eth-exact-lifting} consequently proves consistency of its main patterns.

We finally verify the EPS input conditions and polynomial-time decoding.
All final entries are positive, including labels with zero raw value.
Main rows remain distinct because their raw values differ by at least one and $H>3$.
Auxiliary rows are distinct because $\kmain\ge2$; their tag $3$ separates them modulo $H$ from main rows with tags $1,2$, even if raw values coincide.
Sort the transformed rows once to obtain $U_1<\cdots<U_k$, set $V_j=B_j$, and retain the map to original row indices.
Because $H\le9k+1$, the raw bounds imply $O(N)$ bits for all final entries and $T_{\rm src}$, with polynomial-time construction per inventory.
For any EPS solution, undo sorting and read the retained main symbols, even when different symbols have the same numerical value.
Appendix~\ref{app:eth-palettes} proves that their variable bits form a satisfying assignment; decoding takes polynomial time.

\section{Construction and Canonical Computation}
\label{app:eth-parameters}

\subsection{Parameters and Demand Cylinder Construction}
\label{app:eth-encoding-strategy}
\label{app:eth-baseline}
\label{app:eth-low-clock}
\label{app:eth-high-targets}
\label{app:eth-probe-construction}
Fix an EPS input and let $0<d<1/10000$.
The baseline and encoded instances share capacities and use $\varepsilon=0$ and $\varepsilon=d^2$, respectively.

All pressures below are normalized, and the three unit-capacity supplies start at $(2,4,6)$ in both instances.
Let $X=p_1$, $Z=p_1+p_2$, and $E=X+2Z=3p_1+2p_2$.
We determine the demand capacities from the desired block response in~\eqref{eq:eth-packet-recurrence}: label perturbations accumulate in $X$ as the required row-weighted sum, while zero deviation of $E$ from the baseline is preserved.
We then specify slot and nonterminal clock targets through a formal baseline trajectory, and calibrate the final clock and probe at a reference state corresponding to the EPS equality.
Physical feasibility is not assumed in this design.
Appendix~\ref{app:eth-perturbation} proves that baseline canonical words succeed and that encoded canonical words proceed through the final clock's first two contacts, where the designed response applies.

\subparagraph*{Row multipliers.}
Let $U_{k+1}=U_k+1$ be a formal endpoint, which adds no EPS row or demand.
For $1\le i\le k$, define the desired row multipliers by
\[
 \lambda_i=\frac{1+dU_i}{1+dU_{i+1}},\qquad
 \prod_{h=i}^k\lambda_h=\frac{1+dU_i}{1+dU_{k+1}}.
\]
The product identity follows by cancellation and is affine in $U_i$ with a common denominator.
Thus accumulating label perturbations through the remaining blocks produces the required row-weighted sum.

\subparagraph*{Label capacity and pressures.}
The label parameters control both the perturbation at the first equilibration and the change in total gas requirement.
Let $r_0$ be the common label capacity, chosen to retain coefficient $1-d^2$ at equilibration with a unit supply:
\[
 r_0=\frac{d^2}{1-d^2},\qquad \frac1{1+r_0}=1-d^2.
\]
For label $B_j$, set the initial and target pressures to $1-\varepsilon V_j$ and $3-\tfrac32\varepsilon V_j$, respectively.
The target exceeds the initial pressure by $2-\varepsilon V_j/2>0$ for sufficiently small $d>0$.

The completed label fill requires
\begin{equation}
 r_0\left[\left(3-\frac32\varepsilon V_j\right)-(1-\varepsilon V_j)\right]
 =2r_0-\frac{\varepsilon r_0V_j}{2}.
 \label{eq:eth-label-requirement}
\end{equation}

\subparagraph*{Slot and clock capacity design.}
We choose the slot and clock coefficients to produce the row multiplier and cancel the label contribution to the $E$-deviation.
Slots start at pressure $1$, and clocks start at pressure $0$.
Let $a_i,t_i$ be the slot and clock coefficients to be chosen, with $b_i=1-a_i$ and $\eta_i=1-t_i$; their capacities are $b_i/a_i$ and $\eta_i/t_i$.
Along the intended route, the label and slot equilibrate with supply 1 and complete at supply 1 or 2; the clock equilibrates with supplies 1 and 2.
This is a calculation along the intended route, without assuming that the constructed instance follows it.
Let $\delta X,\delta Z$ be the deviations from the baseline immediately before block $i$, whose label is $B_j$.
Figure~\ref{fig:eth-block-response} shows the local deviation maps along the intended route.
The label and slot draw only from supplies 1 and 2, decreasing their total gas amount $Z$ by the demand requirements.
Thus~\eqref{eq:eth-label-requirement} adds $\varepsilon r_0V_j/2$ to the baseline $Z$-update, while the slot requirement is unchanged.

\par\smallskip
\noindent\begin{minipage}{\linewidth}
\centering
\begin{tikzpicture}[x=1cm,y=1cm,font=\small,
  stage/.style={draw,rounded corners=1pt,minimum height=1.95cm,inner sep=0pt},
  flow/.style={->,semithick}]
  \node[stage,minimum width=5.4cm] (label) at (0,0) {};
  \node at (0,.65) {Label $B_j$};
  \node at (0,-.21) {
    $\begin{aligned}
      \delta X&\mapsto(1-d^2)\delta X-\varepsilon d^2V_j,\\[5pt]
      \delta Z&\mapsto\delta Z+\frac{\varepsilon r_0V_j}{2}.
    \end{aligned}$};
  \node[stage,minimum width=2.7cm] (slot) at (4.45,0) {};
  \node at (4.45,.65) {Slot $A_i$};
  \node at (4.45,-.21) {
    $\begin{aligned}
      \delta X&\mapsto a_i\delta X,\\[5pt]
      \delta Z&\mapsto\delta Z.
    \end{aligned}$};
  \node[stage,minimum width=4.2cm] (clock) at (8.3,0) {};
  \node at (8.3,.65) {Clock $T_i$};
  \node at (8.3,-.21) {
    $\begin{aligned}
      \delta X&\mapsto t_i\delta X,\\[5pt]
      \delta Z&\mapsto t_i\delta Z+t_i\eta_i\delta X.
    \end{aligned}$};
  \draw[flow] (label.east) -- (slot.west);
  \draw[flow] (slot.east) -- (clock.west);
\end{tikzpicture}
\captionsetup{hypcap=false}

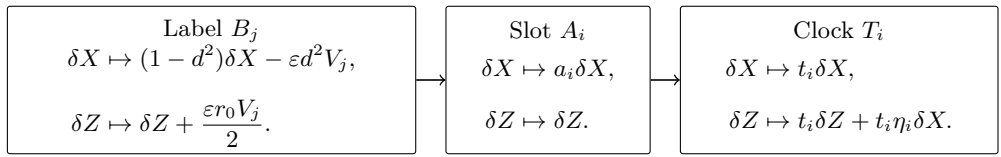
\captionof{figure}{Each box acts simultaneously on its incoming deviations from the baseline at that stage; the clock map ends after its first two equilibrations.}
\label{fig:eth-block-response}
\end{minipage}
\par\smallskip

Let $\delta X^+,\delta E^+$ be the deviations in $X,E$ after the clock's first two equilibrations, with $\delta E=\delta X+2\delta Z$ evaluated at the block entrance.
Composing the maps yields the required row multiplier and $\delta E^+=t_i\delta E$ when
\[
 a_i(1-d^2)t_i=\lambda_i,\qquad a_i(1-d^2)(1+2\eta_i)=1.
\]
The first identity fixes the $X$-multiplier, and the second makes the incoming deviations contribute $t_i\delta E$ to $\delta E^+$.
The second also implies $a_id^2(1+2\eta_i)=r_0$, cancelling the label term in $\delta E^+$.
Using $\eta_i=1-t_i$, these identities determine
\[
 a_i=\frac{1+2\lambda_i}{3(1-d^2)},\qquad b_i=1-a_i,\qquad
 t_i=\frac{3\lambda_i}{1+2\lambda_i},\qquad \eta_i=1-t_i.
\]
The resulting block response is
\begin{equation}
 \delta X^+=\lambda_i(\delta X-\varepsilon r_0V_j),\qquad
 \delta E^+=t_i\delta E.
 \label{eq:eth-packet-recurrence}
\end{equation}
In particular, a zero incoming $E$-deviation remains zero.
Let $\omega_i=(U_{i+1}-U_i)/3>0$.
To check capacity positivity, expand at the fixed EPS input as $d\downarrow0$:
\[
 \lambda_i=1-3d\omega_i+O(d^2),\qquad
 \eta_i=d\omega_i+O(d^2),\qquad b_i=2d\omega_i+O(d^2).
\]
Because $\omega_i>0$, sufficiently small $d>0$ satisfies $0<a_i,t_i<1$.
Thus $b_i/a_i$ and $\eta_i/t_i$ are positive capacities.

\subparagraph*{Formal baseline targets.}
With the capacities fixed, we define a formal baseline trajectory to specify the slot and nonterminal clock targets without assuming physical feasibility.
The \emph{full arrival} at a contact is the demand pressure after equilibration without target stopping.
Slot and nonterminal clock targets are their baseline full arrivals at supplies 1 and 3, respectively, and remain unchanged under encoding.
Let $(x_1,Y_1,W_1)=(1,4,6)$, representing the intended supply state $(1+x_i,Y_i,W_i)$ before block $i$.
After the label's first equilibration, supply 1 has pressure $1+(1-d^2)x_i$.
Set the slot target to $q_{A,i}=1+a_i(1-d^2)x_i,$ its first full arrival, leaving no baseline slot requirement for supply 2.
The label removes $2r_0-d^2x_i$ from supply 2.
The intended clock equilibrations then define
\begin{equation}
 \begin{aligned}
 x_{i+1}=\lambda_i x_i-\eta_i, \quad
 Y_{i+1}=t_i(Y_i-2r_0+d^2x_i)+\eta_i(1+x_{i+1}),\quad
 W_{i+1}=t_iW_i+\eta_iY_{i+1}.
 \end{aligned}
 \label{eq:eth-clock-recursion}
\end{equation}
These recursions define rational functions independent of $\varepsilon$.
For $i<k$, set $q_{T,i}=W_{i+1}$, the intended third full arrival.

\subparagraph*{EPS reference state and final clock target.}
We now calibrate the final clock at a reference supply state $C$ corresponding to the EPS equality.
Let $S=\sum_jV_j$.
Starting with zero deviation, the local response and telescoping row multipliers predict a final $X$-deviation of
\[
 -\frac{\varepsilon r_0}{1+dU_{k+1}}
       \left(S+d\sum_iU_iV_{\pi(i)}\right)
\]
along the intended route.
Replacing the weighted sum by $T_{\rm src}$ defines the target readout
\[
 \theta=-\frac{\varepsilon r_0(S+dT_{\rm src})}{1+dU_{k+1}}=O(d^4)
 \quad\text{at }\varepsilon=d^2.
\]
The target $X$-readout and baseline $E$-value fix the first two coordinates of $C=(C_1,C_2,C_3)$.
In the reference gas balance, supply 3 and the final clock are both at $C_3$, with combined capacity $1+\eta_k/t_k=1/t_k$.
By~\eqref{eq:eth-label-requirement}, all other source requirements equal their baseline values except for the total label saving $\varepsilon r_0S/2$.
Thus define $C$ by
\[
 \begin{aligned}
 C_1&=1+x_{k+1}+\theta,\\
 3C_1+2C_2&=3(1+x_{k+1})+2Y_{k+1},\\
 C_1+C_2+\frac{C_3}{t_k}
   &=1+x_{k+1}+Y_{k+1}+\frac{W_{k+1}}{t_k}
     +\frac{\varepsilon r_0S}{2}.
 \end{aligned}
\]
These equations uniquely determine $C$.
Explicitly,
\[
 C_2=Y_{k+1}-\frac32\theta,\qquad
 C_3=W_{k+1}+\frac{t_k}{2}(\theta+\varepsilon r_0S).
\]
No reachability assumption is used in this definition.
Set $q_{T,k}=C_3$.
The potential equation fixes $E(C)$ to its baseline value.
Subtracting the final clock's requirement $(\eta_k/t_k)C_3$ from the combined gas equation fixes the completed-source supply total to $C_1+C_2+C_3$, independently of source order.
At $d=0$, the formal trajectory is $(1+x_i,Y_i,W_i)=(2,4,6)$.
Its rational functions and the reference coordinates are regular at zero.
At $\varepsilon=0$, $C=(1+x_{k+1},Y_{k+1},W_{k+1})$; at $\varepsilon=d^2$, all three coordinates differ from these baseline values by $O(d^4)$.
Slot and clock targets tend to $2>1$ and $6>0$ and exceed their initial pressures for small enough $d>0$.

\subparagraph*{Probe parameters.}
We calibrate the probe target to its full third arrival from $C$.
Set $a_G=1/10$ and $b_G=9/10$; the probe has capacity $9$ and starts at pressure $1$.
Define $q_G=a_G(C_3+b_GC_2+b_G^2C_1)+b_G^3$, the probe's full third arrival from $C$.
This target tends to $1.851>1$; thus it exceeds the initial pressure for sufficiently small $d>0$.
All demand parameters are now specified in Table~\ref{tab:eth-demand-construction}.

\par\smallskip
\noindent\begin{minipage}{\linewidth}
\centering\small
\captionsetup{hypcap=false}
\captionof{table}{Demand parameters before the final pressure translation and scaling.}
\label{tab:eth-demand-construction}
\setlength{\tabcolsep}{5pt}
\renewcommand{\arraystretch}{1.15}
\begin{tabularx}{\linewidth}{@{}lcc>{\centering\arraybackslash}X@{}}
\toprule
Demand cylinder & Capacity & Initial pressure & Target pressure\\
\midrule
Label $B$ ($\varepsilon=0$) & $r_0$ & $1$ & $3$\\
Label $B_j^\varepsilon$ ($\varepsilon>0$) & $r_0$ & $1-\varepsilon V_j$ & $3-\tfrac32\varepsilon V_j$\\
Slot $A_i$ & $b_i/a_i$ & $1$ & $q_{A,i}$\\
Clock $T_i$ & $\eta_i/t_i$ & $0$ & $q_{T,i}$\\
Probe $G$ & $9$ & $1$ & $q_G$\\
\bottomrule
\end{tabularx}
\end{minipage}
\par\smallskip
Call labels and slots \emph{low demands}, and a repetition-free sequence of source demands a \emph{source prefix}.

\subsection{Common Contact Routes and the Terminal Bound}
\label{app:eth-common-routes}
For both instances, fix an EPS input and sufficiently small $d>0$.
The \emph{arrival gap} is the full-equilibration pressure minus the target; a demand completes at a contact exactly when this gap is nonnegative.
Source capacities and total requirements are $O(d)$ because label, slot, and clock targets tend to $3$, $2$, and $6$.
Withdrawal during any successful source prefix or partial source fill is bounded by this total; hence $(p_1,p_2,p_3)=(2,4,6)+O(d)$ uniformly.
These separated limits establish the source routes in Table~\ref{tab:encoded-transfer-branches}, with positive flow at every used contact.
The probe's first two full arrivals tend to $1.1$ and $1.39$, below its limiting target $1.851$, with positive flow.
Completion at supply 3 remains a separate arrival-gap test.

\par\smallskip
\noindent\begin{minipage}{\linewidth}
\centering\small
\captionsetup{hypcap=false}
\captionof{table}{Contact routes after a successful source prefix.}
\label{tab:encoded-transfer-branches}
\begin{tabularx}{\linewidth}{@{}l>{\raggedright\arraybackslash}X>{\raggedright\arraybackslash}Xl@{}}
\toprule
Demand cylinder & Cylinder 1 & Cylinder 2 & Cylinder 3\\
\midrule
Label $B_j$ & Equilibration & Target stopping & Unused\\
Slot $A_i$ & Equilibration or target stopping & Target stopping if unfinished & Unused\\
Clock $T_i$ or probe $G$ & Equilibration below target & Equilibration below target & Target stopping on success\\
\bottomrule
\end{tabularx}
\end{minipage}
\par\smallskip

\phantomsection
\label{app:eth-shared-terminal}
Let $\delta X=p_1-C_1$ and $\delta Z=p_1+p_2-C_1-C_2$ for a completed source state $p$ with the same total supply pressure as $C$.
Conservation implies $p_3-C_3=-\delta Z$.
The probe's full third arrival weights the pressure deviations of supplies 1, 2, and 3 by $a_Gb_G^2$, $a_Gb_G$, and $a_G$, respectively.
After the probe's first two equilibrations below target, subtracting $q_G$ from its full third arrival yields $-a_G^2(\delta Z+b_G\delta X)$.
Thus $p_3\ge C_3$ and probe completion require $\delta Z\le0$ and $\delta Z+b_G\delta X\le0$, whose intersection appears in Figure~\ref{fig:eth-terminal-region}.
Because $b_G>1/2$,
\[
 E(p)-E(C)=\delta X+2\delta Z
 =\left(2-\frac1{b_G}\right)\delta Z+\frac1{b_G}(\delta Z+b_G\delta X)\le0.
\]

\par\smallskip
\noindent\begin{minipage}{\linewidth}
\centering
\begin{tikzpicture}[x=1.1cm,y=.85cm,font=\small]
  \fill[black!12] (-2.4,-2.35) -- (2.4,-2.35) -- (2.4,-2.16) -- (0,0) -- (-2.4,0) -- cycle;
  \draw[->,thin] (-2.55,0) -- (2.62,0) node[anchor=west] {$\delta X$};
  \draw[->,thin] (0,-2.49) -- (0,1.32) node[anchor=south] {$\delta Z$};
  \draw[semithick] (-2.4,0) -- (2.4,0);
  \draw[semithick] (-1.3,1.17) -- (2.4,-2.16);
  \draw[dashed,semithick] (-2.3,1.15) -- (2.4,-1.2);
  \fill (0,0) circle[radius=1.7pt];
  \node[anchor=north east,fill=white,inner sep=1pt] at (-.08,-.08) {$C$};
  \node at (-1.22,-1.34) {$E(p)<E(C)$};
  \node at (1.13,.81) {$E(p)>E(C)$};
  \draw[semithick] (3.3,.68) -- (3.8,.68);
  \node[anchor=west,align=left] at (3.95,.68) {Condition $p_3\ge C_3$\\[2pt]$\delta Z\le0$};
  \draw[semithick] (3.3,-.44) -- (3.8,-.44);
  \node[anchor=west,align=left] at (3.95,-.44) {Probe condition\\[2pt]$\delta Z+b_G\delta X\le0$};
  \draw[dashed,semithick] (3.3,-1.56) -- (3.8,-1.56);
  \node[anchor=west,align=left] at (3.95,-1.56) {Potential boundary\\[2pt]$\delta X+2\delta Z=0$};
\end{tikzpicture}
\captionsetup{hypcap=false}

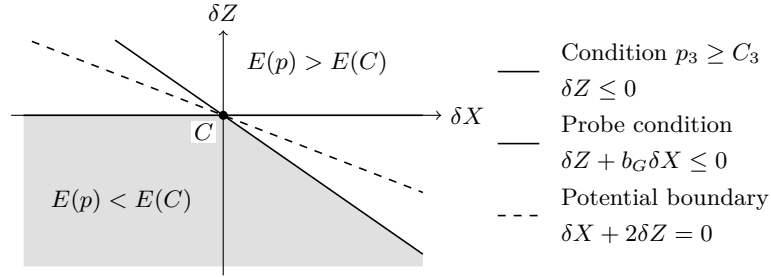
\captionof{figure}{Necessary terminal conditions for $b_G=9/10$; the shaded region meets $E(p)=E(C)$ only at $C$.}
\label{fig:eth-terminal-region}
\end{minipage}
\par\smallskip

\subsection{Canonical Feasibility and Endpoint Response}
\label{app:eth-perturbation}
\phantomsection
\subparagraph*{Canonical feasibility.}
\label{app:eth-canonical-feasibility}
We first show that every baseline canonical word succeeds and reaches $C$ before the probe, and that every encoded canonical word proceeds through the first two contacts of $T_k$.
Fix an EPS input and let $d>0$ be sufficiently small.

Set $\varepsilon=0$ and induct on canonical blocks from the initial state.
Starting from the formal state before block $i$, the label equilibrates with supply 1 and removes $2r_0-d^2x_i$ from supply 2.
The slot reaches $q_{A,i}$ exactly at its first equilibration.
The clock's first two equilibrations yield $(1+x_{i+1},Y_{i+1})$, and its third full arrival is $W_{i+1}=q_{T,i}$, including $i=k$ because $C_3=W_{k+1}$.
Thus every baseline block completes at the next formal state, which proves that the formal trajectory is physically realized.
The final state is $C$, and the calibrated probe's full third arrival equals $q_G$.
The common routes therefore imply its completion.

Set $\varepsilon=d^2$ and let $\delta X_i,\delta Z_i,\delta E_i$ be deviations from the baseline after the first two contacts of $T_i$, with $\delta X_0=\delta Z_0=\delta E_0=0$.
The third contact changes neither $X$ nor $Z$.
Suppose the preceding blocks have completed with $\delta X_{i-1}\le0$ and $\delta E_{i-1}=0$, and consider label $B_{\pi(i)}$.
The label follows its common route.
Its output supply-1 pressure differs from the baseline slot threshold $1+(1-d^2)x_i$ by $(1-d^2)\delta X_{i-1}-\varepsilon d^2V_{\pi(i)}<0$.
Thus the slot equilibrates with supply 1 below target and completes at supply 2.
The clock's first two contacts also follow the common route.
The local response~\eqref{eq:eth-packet-recurrence} now applies to these physical contacts and yields
\[
 \delta X_i=\lambda_i(\delta X_{i-1}-\varepsilon r_0V_{\pi(i)})<0,
 \qquad \delta E_i=0.
\]
The telescoping multipliers imply the exact prefix response
\begin{equation}
 \delta X_i=-\frac{\varepsilon r_0}{1+dU_{i+1}}
                  \sum_{h\le i}(1+dU_h)V_{\pi(h)},\qquad
 \delta Z_i=-\frac{\delta X_i}{2}.
 \label{eq:eth-canonical-prefix-response}
\end{equation}
For $i<k$, the clock targets retain their baseline values, whereas~\eqref{eq:eth-label-requirement} shows that the completed labels save $(\varepsilon r_0/2)\sum_{h\le i}V_{\pi(h)}$ in total gas.
After the first two contacts of $T_i$, the remaining gas difference belongs to supply 3 and the unfinished clock, whose combined capacity is $1/t_i$.
Their full third arrival $h_{3,i}$ therefore satisfies
\[
 \begin{aligned}
 h_{3,i}-q_{T,i}
 =\frac{t_i}{2}\left(\varepsilon r_0\sum_{h\le i}V_{\pi(h)}+\delta X_i\right)
 =\frac{t_i\varepsilon r_0d}{2(1+dU_{i+1})}
                  \sum_{h\le i}(U_{i+1}-U_h)V_{\pi(h)}>0.
 \end{aligned}
\]
Every summand is positive because the row weights increase strictly and the label values are positive.
This proves completion of each nonterminal clock and closes the induction through the first two contacts of $T_k$, without assuming its completion.
If $T_k$ completes, the common routes apply to the subsequent probe.

\phantomsection
\subparagraph*{Canonical endpoint response.}
\label{app:eth-endpoint-algebra}
With $\varepsilon=d^2$ and sufficiently small $d>0$, equation~\eqref{eq:eth-canonical-prefix-response} and the definition of $\theta$ yield, after the first two contacts of $T_k$,
\[
 \begin{aligned}
 p_1-C_1
 =\delta X_k-\theta
 =\frac{\varepsilon r_0d}{1+dU_{k+1}}
       \left(T_{\rm src}-\sum_iU_iV_{\pi(i)}\right)
 =:x_\pi.
 \end{aligned}
\]
Here $\alpha=\varepsilon r_0d/(1+dU_{k+1})>0$ is common to all permutations.
The local response preserves zero $E$-deviation from the baseline, and $C$ has the baseline potential.
Hence $3(p_1-C_1)+2(p_2-C_2)=0$, implying $p_2-C_2=-3x_\pi/2$.
The final clock's third contact changes neither of these pressures.

\section{Order Certificates and Completion of the Construction}
\label{app:eth-clock}
A \emph{source word} lists all $3k$ source demands.
All asymptotic constants below may depend on the fixed EPS input, and the bounds are uniform over words and prefixes.

\subsection{Physical Order and the Terminal Bound}
\label{app:eth-physical-order}
Let $E(p)=3p_1+2p_2$.
Every successful full encoded word satisfies three necessary conditions:
\begin{enumerate}[(i)]
\item The probe $G$ is the last demand.
\item The clocks occur in index order.
\item The supply state $p$ immediately before $G$ satisfies $E(p)\le E(C)$.
\end{enumerate}
Appendix~\ref{app:eth-endpoint-certificates} uses this bound to exclude noncanonical source orders with clocks in index order.
Fix an EPS input and set $\varepsilon=d^2$, with $d>0$ sufficiently small.

\subparagraph*{(i) Probe last.} The probe $G$ must be last because completing it lowers all supply pressures below every remaining source target.
By the common source-prefix estimate and routes in Appendix~\ref{app:eth-common-routes}, a successful probe leaves the first two supplies at pressures tending to $1.1$ and $1.39$ as $d\downarrow0$.
Target stopping leaves the third at a pressure tending to $6-9(1.851-1.39)=1.851$.
All three limits are below $1.9$.
Every remaining source demand has initial pressure tending to at most one and target tending to at least two.
For sufficiently small $d$, these targets exceed all residual supply pressures and all source initial pressures; none can finish afterward.

\subparagraph*{(ii) Clock order.} The clocks must occur in index order.
The target of $T_i$ accounts for withdrawals by $T_1,\ldots,T_i$; any additional clock consumes pressure needed for $T_i$.
A successful clock $T_j$ withdraws $2\omega_jd+O(d^2)$ from supply 3: its capacity is $\omega_jd+O(d^2)$, and its target and pressure after the second equilibration tend to $6$ and $4$.
Low demands leave this supply unchanged, and it starts at $6$.
Let $W(J)$ be the pressure of supply 3 when the set of completed clocks is $J$.
Then
\[
 W(J)=6-2d\sum_{j\in J}\omega_j+O(d^2).
\]
Baseline calibration fixes $q_{T,i}=6-2d\sum_{j\le i}\omega_j+O(d^2)$.
The final-clock target change is $O(d^4)$; thus this expansion also holds under encoding.
Every successful $T_i$ leaves supply 3 at least $q_{T,i}$.
At the first departure from index order, let $T_i$ be the next unused clock.
At this departure, a clock $T_j$ with $j>i$ completes before $T_i$.
When $T_i$ eventually completes, let $J$ be the set of all clocks completed by that moment.
The set $J$ contains $\{1,\ldots,i\}$ and at least the additional index $j$.
Consequently
\[
 W(J)-q_{T,i}=-2d\sum_{j\in J\setminus\{1,\ldots,i\}}\omega_j+O(d^2)<0.
\]
Because the additional clocks have positive weights, their negative contribution of order $d$ dominates the $O(d^2)$ error for sufficiently small $d$.
This violates the requirement that supply 3 have pressure $\ge q_{T,i}$ when $T_i$ completes.

\subparagraph*{(iii) Terminal bound.} The bound follows by combining total-pressure conservation with the feasibility of $T_k$ and $G$.
Let $p=(p_1,p_2,p_3)$ be the supply state immediately before $G$.
The gas-budget identity in Appendix~\ref{app:eth-baseline} fixes the total supply pressure exactly to that of $C$.
Completion of $T_k$ leaves supply 3 at least $q_{T,k}=C_3$, and later low demands leave it unchanged; hence $p_3\ge C_3$.
Let $\delta X=p_1-C_1$ and $\delta Z=p_1+p_2-C_1-C_2$.
Conservation and $p_3\ge C_3$ imply $\delta Z=-(p_3-C_3)\le0$.
The probe's first two contacts are the equilibrations established above.
Its success therefore implies $\delta Z+b_G\delta X\le0$ by the arrival-gap identity in Appendix~\ref{app:eth-shared-terminal}.
Because $b_G>1/2$, the potential difference is a positive linear combination of these two nonpositive quantities:
\[
 E(p)-E(C)
 =\left(2-\frac1{b_G}\right)\delta Z
  +\frac1{b_G}(\delta Z+b_G\delta X)\le0.
\]

\subsection{Baseline Order Penalties and the Encoding Error}
\label{app:eth-endpoint-certificates}
We show that the encoded endpoint potential of every noncanonical source order with clocks in index order exceeds $E(C)$ for sufficiently small $d>0$.
We first detect departures from the canonical slot--clock order through a positive $d^2$ term in the baseline endpoint difference.
If this coefficient vanishes, we detect departures from the canonical label placement through a positive $d^3$ term.
We then bound the encoding error uniformly by $O(d^4)$, which preserves the positive leading term in either case.

\subparagraph*{Baseline and encoded comparison maps.}
We define algebraic comparison maps to track the endpoint potential in both instances.
The baseline below is an algebraic comparison program for the same source order; its physical execution need not succeed.
Its canonical endpoint potential is $E(C)$ even though $C$ denotes the desired encoded endpoint.
Let $X=p_1$, $Z=p_1+p_2$, and $E=X+2Z$.
Define $H_i=1+(1-d^2)x_i$ and $\slotreq i=b_i(H_i-1)$.
Here $H_i$ is the incoming supply-1 pressure whose first slot arrival equals the target, and $\slotreq i$ is the slot's total requirement.
With $(u)_+=\max\{u,0\}$, the baseline reduced maps are
\begin{align*}
 B:\quad &(X,Z)\mapsto((1-d^2)X+d^2,\ Z-2r_0),\\
 A_i:\quad &(X,Z)\mapsto(X-\slotreq i+b_i(H_i-X)_+,\ Z-\slotreq i),\\
 T_i:\quad &(X,Z)\mapsto(t_iX,\ t_iZ+t_i\eta_iX).
\end{align*}
Only the label map on $(X,Z)$ changes under encoding: $B_j^\varepsilon(X,Z)=B(X,Z)-\varepsilon V_j\,(d^2,-r_0/2).$
Both algebraic programs remain defined for every repetition-free source prefix, even when physical filling fails.
Every update branch is the identity at $d=0$; hence every algebraic source prefix has $(X,Z)=(2,6)+O(d)$, uniformly at the fixed input.
Along successful encoded source prefixes, the encoded program agrees with physical filling by Appendix~\ref{app:eth-common-routes}.
\begin{itemize}
\item For $X\ge H_i$, supply 1 alone fills slot $i$.
\item For $X<H_i$, the positive-part correction is the portion of its requirement supplied by supply 2 instead.
\end{itemize}
Equilibration with the label subtracts $d^2\varepsilon V_j$ from $X$, whereas its total requirement $r_0(2-\varepsilon V_j/2)$ adds $r_0\varepsilon V_j/2$ to $Z$ relative to baseline.
At the fixed EPS input, the label pressure shifts are $O(d^2)$ and $r_0=O(d^2)$; hence each encoded label changes the supply-state map by $O(d^4)$.
For a complete source word $w$, let $E_0(w)$ and $E_\varepsilon(w)$ denote the endpoint potentials of these baseline and encoded programs.
The following calculation concerns $E_0(w)$, regardless of baseline physical feasibility.
The remainders are uniform over words at the fixed input.

\subparagraph*{Accounting for the remaining clocks.}
We account for the remaining clocks first to compare order differences through low-demand contributions alone.
To isolate these contributions, we incorporate all future clock maps into the potential, making its value unchanged across each clock contact.
Fix a complete source word $w$ whose clocks occur in index order.
After the first $j$ clocks, use the linear potential $E\circ T_k\circ\cdots\circ T_{j+1}$.
At a clock contact, the value before its reduced map equals the value after that map with the clock removed from the suffix.
The exact remaining-clock map is
\[
 (E\circ T_k\circ\cdots\circ T_{j+1})(X,Z)
 =\left(\prod_{h=j+1}^k t_h\right)
   \left[\left(1+2\sum_{h=j+1}^k\eta_h\right)X+2Z\right].
\]
Indeed, composing the maps $(X,Z)\mapsto t_h(X,Z+\eta_hX)$ multiplies their scalar factors and adds their shear coefficients $\eta_h$.
For $j=k$, the empty product is one and the empty sum is zero, recovering $E$; for $j=0$, the formula includes every clock.
All words start with the same value, and the value at the endpoint is $E$.
A low-demand \emph{jump} is this potential's value after its update minus its value before. Endpoint differences equal differences of summed jumps, using each word's pending clocks and slot branches.

Let $\rho_i=\sum_{r<i}\omega_r$.
For a low demand $D$, let $P_D$ and $Q_D$ be the sums of construction weights $\omega_j$ over preceding slots $A_j$ and clocks $T_j$, respectively.
Expanding the exact suffix formula with $t_h=1-d\omega_h+O(d^2)$ and $\eta_h=d\omega_h+O(d^2)$ yields coefficients $(1,2)+d(\sum_h\omega_h-Q_D)(1,-2)+O(d^2)$ at $D$.
For slot and label supply-1 inputs $X_i^-$ and $X_B$, set $\delta_i=b_i(H_i-X_i^-)_+$.
Their respective $(X,Z)$ changes are $(-\slotreq i+\delta_i,-\slotreq i)$ and $(-d^2(X_B-1),-2r_0)$.

\subparagraph*{Slot and clock order.}
We derive~\eqref{eq:eth-dual-certificate} and show that its $d^2$ coefficient vanishes exactly when the slot--clock projection is $A_1T_1\cdots A_kT_k$.
The canonical recursion~\eqref{eq:eth-clock-recursion} yields $H_i=2-4d\rho_i+O(d^2)$.
Every slot or clock lowers $X$ by $2d$ times its weight, up to $O(d^2)$, whereas labels change $X$ by $O(d^2)$.
Let $\Delta_i=P_{A_i}+Q_{A_i}-2\rho_i$.
This quantity measures the excess or deficit in the total weight of slots and clocks preceding $A_i$ relative to the canonical value $2\rho_i$.
Thus
\[
 H_i-X_i^-=2d\Delta_i+O(d^2),\qquad
 \delta_i=4d^2\omega_i(\Delta_i)_++O(d^3).
\]
The positive-part estimate includes zero by Lipschitz continuity.
The slot requirement $\slotreq i=2d\omega_i+O(d^2)$ is position-independent, and its canonical correction is zero.
Subtracting its canonical jump therefore leaves $\delta_i+2d^2\omega_i(\rho_i-Q_{A_i})+O(d^3)$.
Labels' common leading jumps $-d^2-4r_0$ cancel.
In $\sum_i\omega_iP_{A_i}$, each unordered slot pair contributes its weight product once.
Thus $\sum_i\omega_i(P_{A_i}-\rho_i)=0$.
Hence $\sum_i\omega_i(\rho_i-Q_{A_i})=-\sum_i\omega_i\Delta_i$.
Summing the slot and label contributions and applying this identity with $2x_+-x=|x|$ yields
\begin{equation}
 \begin{aligned}
 E_0(w)-E(C)
 =2d^2\sum_i\omega_i
   \bigl[2(\Delta_i)_++\rho_i-Q_{A_i}\bigr]+O(d^3)
 =2d^2\sum_i\omega_i|\Delta_i|+O(d^3).
 \end{aligned}
 \label{eq:eth-dual-certificate}
\end{equation}
The positive-part correction penalizes excess preceding weight; the linear term supplies the matching penalty for a deficit.
The totals $P_{A_i}+Q_{A_i}$ increase strictly in slot execution order.
If every $\Delta_i=0$, they equal the strictly increasing values $2\rho_i$, forcing slot index order.
Then $P_{A_i}=\rho_i$ and $Q_{A_i}=\rho_i$.
Ordered clocks and positive weights force exactly $A_1T_1\cdots A_kT_k$.

\subparagraph*{Label placement.}
With the projection after deleting labels fixed to $A_1T_1\cdots A_kT_k$, we derive~\eqref{eq:eth-label-placement-certificate} and show that its $d^3$ coefficient vanishes exactly when one label immediately precedes each pair $A_iT_i$.
Let $N_i$ count labels before $A_i$ and $\ell_i$ those between $A_i$ and $T_i$.
Each label lowers $X$ by $d^2+O(d^3)$.
At corresponding nonlabel events, the incoming $X$ values in the baseline executions of $w$ and a canonical word differ by $O(d^2)$; a common slot or clock changes this difference only by $O(d^3)$, because the slot map minus the identity is $b_i$-Lipschitz and the clock slope is $1+O(d)$.
Comparison with the canonical input $H_i$ yields
\[
 H_i-X_i^-=(N_i-i)d^2+O(d^3),\qquad
 \delta_i=2d^3\omega_i(N_i-i)_++O(d^4).
\]
Each slot now has the same pending clocks as in the canonical word; hence its jump difference is $(1+O(d))\delta_i=\delta_i+O(d^4)$.
For a label, $X_B-1=1-2d(P_B+Q_B)+O(d^2)$; hence its jump is
\[
 -d^2-4r_0+3d^3\sum_h\omega_h+d^3(2P_B-Q_B)+O(d^4).
\]
The first three terms cancel in the canonical subtraction.
Slot $i$ precedes $k-N_i$ labels, and clock $i$ precedes $k-N_i-\ell_i$ labels.
Hence $\sum_B(2P_B-Q_B)=\sum_i\omega_i(k-N_i+\ell_i)$.
Subtracting canonical values and adding the slot jumps, with $2x_+-x=|x|$, proves
\begin{equation}
 E_0(w)-E(C)=d^3\sum_i\omega_i\bigl(|N_i-i|+\ell_i\bigr)+O(d^4).
 \label{eq:eth-label-placement-certificate}
\end{equation}
Its coefficient vanishes exactly when $N_i=i$ and $\ell_i=0$ for every $i$: one label immediately before each pair and none inside or after the pairs.
Figure~\ref{fig:eth-order-examples} illustrates the two penalties.

\par\smallskip
\noindent\begin{minipage}{\linewidth}
\centering
\begin{tikzpicture}[x=.66cm,y=.66cm,font=\small]
 \node[anchor=east] at (-.6,0) {Canonical};
 \node[anchor=east] at (-.6,-1) {Slot $A_2$ too early};
 \node[anchor=east] at (-.6,-2) {Extra label before $A_1$};
 \foreach \x/\s in {0/B,1/A_1,2/T_1,3/B,4/A_2,5/T_2}
   \node[draw,rounded corners=1pt,minimum width=.55cm,minimum height=.42cm] at (\x,0) {$\s$};
 \foreach \x/\s in {0/B,1/A_1,2/B,3/A_2,4/T_1,5/T_2}
   \node[draw,rounded corners=1pt,minimum width=.55cm,minimum height=.42cm] at (\x,-1) {$\s$};
 \foreach \x/\s in {0/B,1/B,2/A_1,3/T_1,4/A_2,5/T_2}
   \node[draw,rounded corners=1pt,minimum width=.55cm,minimum height=.42cm] at (\x,-2) {$\s$};
 \draw[->] (5.7,0) -- (6.5,0);
 \draw[->] (5.7,-1) -- (6.5,-1);
 \draw[->] (5.7,-2) -- (6.5,-2);
 \node[anchor=west] at (6.7,0) {$E_0(w)-E(C)=0$};
 \node[anchor=west] at (6.7,-1) {$2d^2\omega_1\omega_2+O(d^3)$};
 \node[anchor=west] at (6.7,-2) {$d^3\omega_1+O(d^4)$};
\end{tikzpicture}
\captionsetup{hypcap=false}
\captionof{figure}{Two-block examples in the baseline reduced program; each $B$ denotes a label.
In the second row, $\Delta_2=-\omega_1$.
In the third, $N_1-1=1$ and all other label-position deviations vanish.}
\label{fig:eth-order-examples}
\end{minipage}
\par\smallskip
\subparagraph*{Encoding error and order exclusion.}
\phantomsection
\label{app:eth-prefix-estimate}
We show that encoding changes every prefix state by only $O(d^4)$ and hence preserves the positive leading penalty for each noncanonical order.
For corresponding encoded and baseline iterates, with $d^2\le1/2$, $\varepsilon=d^2$, and $S=\sum_jV_j$,
\begin{equation}
 \|(X,Z)_\varepsilon-(X,Z)_0\|_\infty
 \le \varepsilon d^2S=O(d^4).
 \label{eq:eth-prefix-bound}
\end{equation}
Because $E=X+2Z$, this implies $|E_\varepsilon-E_0|\le3\varepsilon d^2S=O(d^4)$.
These estimates are uniform over words and prefixes, with constants depending on the fixed EPS input.
For sufficiently small $d$, Appendix~\ref{app:eth-encoding-strategy} has $0<a_i,t_i,1-d^2<1$.
Each $(X,Z)$ map is nonexpansive in the infinity norm.
\begin{itemize}
\item The label's linear part is $\operatorname{diag}(1-d^2,1)$.
\item The slot's $X$-slopes are $a_i$ and $1$, its $Z$-update is a translation independent of $X$, and continuity at the breakpoint extends the bound across branches.
\item The clock's $(X,Z)$ linear part has nonnegative coefficients and row sums $t_i<1$ and $t_i(1+\eta_i)=1-\eta_i^2\le1$.
\end{itemize}
The initial difference is zero.
Each encoded label adds $-\varepsilon V_j(d^2,-r_0/2)$, of norm at most $\varepsilon d^2V_j$ because $r_0/2\le d^2$.
Induction proves~\eqref{eq:eth-prefix-bound}, including when slot branches differ; $\varepsilon=d^2$ yields the stated order.
For the encoded reduced endpoint, \eqref{eq:eth-prefix-bound} adds only $O(d^4)$ to each certificate.
The first $d^2$ coefficient, and the second $d^3$ coefficient on the required slot--clock projection, are therefore unchanged.
Every noncanonical source order with clocks in index order therefore has a positive leading coefficient in one of these certificates; hence $E_\varepsilon(w)>E(C)$ for sufficiently small $d$.

\subsection{Uniform Validation and an Effective Scale Choice}
\label{app:eth-stability}
\label{app:eth-effective-scale}
We replace the preceding sufficiently-small-$d$ conditions by one rational scale of polynomial bit length, computable in polynomial time, that preserves all required physical and algebraic comparison signs.
Use $\varepsilon=0$ for baseline parameters and $\varepsilon=d^2$ for encoded parameters in both computations.

\subparagraph*{Sign preservation principle.}
For an integer polynomial $F$, let $\|F\|_1$ be the sum of its absolute coefficients.
Suppose $F(d)=d^r(a_r+a_{r+1}d+\cdots)$ with $a_r\ne0$ and $\|F\|_1\le2^{B_*}$ for an integer $B_*\ge14$.
Because the coefficients are integers, $|a_r|\ge1$.
At $d=2^{-B_*-3}$, the higher-degree terms satisfy
\[
 \left|\sum_{j>r}a_jd^{j-r}\right|
 \le d\|F\|_1\le d2^{B_*}=\frac18<|a_r|.
\]
Thus $F(d)$ retains the sign of $a_r$, including when $F$ vanishes at zero.
Applying this principle separately to every nonzero comparison numerator and denominator preserves quotient signs and ensures nonzero denominators, without requiring positive denominator signs; identically zero comparisons stay zero.
It remains to compute a polynomially bounded $B_*$ from the input that controls all words and branches simultaneously, without enumerating them.

\subparagraph*{Physical and algebraic computations.}
We express both computations as affine updates and list the comparisons whose signs must be preserved.

For physical filling use homogeneous coordinates $(p_1,p_2,p_3,h,1)$, initially $(2,4,6,0,1)$.
The final coordinate $1$ allows the constant terms of affine updates to be represented by matrix multiplication.
For each demand of capacity $r$, initial pressure $z$, and target $q$, reset $h=z$ and scan supplies $i=1,2,3$, stopping when $h=q$.
Skip a contact if $p_i\le h$; otherwise compare $e=(p_i+rh)/(1+r)$ with $q$ and update
\[
 (p_i,h)\longmapsto
 \begin{cases}
  (e,e),&e<q,\\
  (p_i-r(q-h),q),&e\ge q.
 \end{cases}
\]
The second branch completes the demand; if $h<q$ after supply 3, stop the physical computation with failure.
Lemma~\ref{lem:cascade-supply-order} preserves the initial supply index order between fills.
\begin{itemize}
\item For $e=q$, both updates yield $(q,q)$.
\item For $p_i=h$, equilibration agrees with skipping.
\end{itemize}
Each reset and contact update is affine with rational-function coefficients in $d$; computing $e$ is a comparison row, not another state update.
Thus every physical word uses at most $K=4(3k+1)$ updates.

Reuse the baseline and encoded maps of Appendix~\ref{app:eth-endpoint-certificates} for algebraic comparison.
Embed $(X,Z,1)$ as $(X,Z,0,0,1)$, initially $(2,6,0,0,1)$.
These algebraic computations use at most $3k\le K$ updates, branch on $H_i-X$, and end with the comparison $X+2Z-E(C)$.
They continue through every source word, even after physical failure; agreement with physical filling is used only on successful source prefixes.
Their slot branches coincide at $X=H_i$.
\begin{itemize}
\item Dynamic comparisons. Include all physical contact, route, and post-probe exhaustion comparisons used in Appendices~\ref{app:eth-common-routes} and~\ref{app:eth-physical-order}, and both algebraic terminal comparisons.
\item Static comparisons. Include comparisons for positive capacities, targets above initial pressures, initial supply order, positive pressures after adding one, $0<a_i,t_i,1-d^2,\lambda_i<1$, and $d^2\le1/2$.
\end{itemize}

\subparagraph*{Polynomial-size local data.}
The triangle inequality and coefficient convolution imply $\|F+G\|_1\le\|F\|_1+\|G\|_1$ and $\|FG\|_1\le\|F\|_1\|G\|_1$.
Substituting successive recursion coordinates into equation~\eqref{eq:eth-clock-recursion} makes calibration $k$ affine updates in fixed dimension, with constant-degree primitive rational coefficients of polynomial coefficient bit length.
The product of their $O(k)$ primitive denominators has polynomial degree and logarithmic norm.
After clearing this denominator, the fixed-dimensional matrix product has polynomial degree and logarithmic coefficient norm: its length is $k$, and each entry sums at most $5^{k+1}$ products of $k+2$ entries.
The explicit coordinate formulas for $C$ in Appendix~\ref{app:eth-baseline} define one further affine map with primitive rational coefficients.
Each target, local update, and comparison uses a bounded number of rational operations on the calibrated coordinates and primitive data.
Consequently, all local data admit expanded integer-polynomial quotients of polynomial degree and coefficient bit length, computable in polynomial time.

\subparagraph*{Uniform bounds over all words.}
Clear component denominators separately in each local matrix, initial column, and comparison row, obtaining one nonzero integer-polynomial denominator per array.
Compute $D\ge1$ as the maximum degree of these numerator components and denominators, and $B\ge1$ by rounding up the logarithm of their maximum coefficient norm.
The local list is polynomial in size and includes both parameter choices.
Every word and branch sequence uses arrays from this same list, and every computation has at most $K$ updates.
Thus the same $D,B,K$ bound all resulting comparison degrees and coefficient norms without enumerating words or branch sequences.
After any $\ell\le K$ updates, a comparison is $P/Q$: the numerator sums at most $5^{\ell+1}$ products of $\ell+2$ components, and the denominator multiplies the corresponding array denominators.
The norm product inequality therefore proves
\[
 \begin{aligned}
 \deg P,\deg Q&\le(K+2)D,\quad
 \|P\|_1\le5^{K+1}2^{(K+2)B},\quad
 \|Q\|_1\le2^{(K+2)B}.
 \end{aligned}
\]

\subparagraph*{Choice of a common scale.}
\phantomsection
\label{app:eth-scale-choice}
We now apply the sign preservation principle to the uniform coefficient bounds.
To ensure this simultaneously for all comparison numerators and denominators, choose
\begin{equation}
 B_*=\max\{14,(K+2)B+3(K+1)\},\qquad
 d=2^{-B_*-3},\quad \varepsilon=d^2.
 \label{eq:eth-scale-choice}
\end{equation}
Keep $\varepsilon=0$ in the baseline.
The expanded local data make $D,B,B_*$ polynomially bounded and computable in polynomial time.
Both polynomial norms are at most $2^{B_*}$, and $B_*\ge14$ implies $d<1/10000$.
The principle preserves the limiting signs of every nonzero numerator and denominator, including denominators vanishing at zero; identically zero comparisons stay zero.
In particular, denominators are nonzero at the chosen scale; their signs need not be positive.
Induction along each computation preserves every comparison and branch, including stopping or failure in the physical computation and all later algebraic updates.
Equalities preserve the coincident updates above.

\subparagraph*{Validation of the construction.}
For every noncanonical source word with clocks in index order, Appendix~\ref{app:eth-endpoint-certificates} proves that $E_\varepsilon(w)-E(C)$ has a positive leading term.
This algebraic endpoint comparison is included above and therefore remains positive at the chosen scale, independently of physical feasibility.
For a successful full encoded word, its source prefix has the same algebraic and physical endpoint; a noncanonical prefix therefore contradicts the terminal bound in Appendix~\ref{app:eth-physical-order}.
The preceding appendices establish the required route, completion, and static signs as $d\downarrow0$; the same scale preserves them and the canonical terminal equalities and tests of Lemma~\ref{lem:eth-canonical-computation}.

\subsection{Output Normalization and Size}
\label{app:eth-output}
At the scale in~\eqref{eq:eth-scale-choice}, evaluating a nonzero polynomial $F$ uses $O(\log(1+\|F\|_1)+(\deg F)(B_*+3))$ bits; hence the calibrated parameters are polynomial-bit rationals computable in polynomial time.
Add one to every initial and target pressure to make them positive, then multiply them by the product of their denominators; leave capacities unchanged.
Common translation preserves pressure differences and transfer amounts, and positive scaling preserves every completion test while scaling transferred amounts.
Thus both operations preserve successful words.
There are $O(k)$ pressures, each of polynomial bit length; the bit length of their denominator product is at most the sum of those lengths.
Thus the normalization is computable in polynomial time.
The resulting instance satisfies the output conditions of Lemma~\ref{lem:eth-physical-realization}.

\end{document}